%% file: main.tex
\documentclass[acmsmall,screen,nonacm]{acmart}
\pdfoutput=1
\input{tex/prelude}

\title{Formal Semantics for the Halide Language}

\begin{document}
\input{tex/abstract}

\maketitle

\section{Introduction}
\label{sec:introduction}
\input{tex/introduction}

\section{An Example Halide Program}
\label{sec:example}
\input{tex/example}

% \clearpage
\section{Overview \& Proof Structure}
\label{sec:overview}
\input{tex/overview}

\section{Algorithm Language}
\label{sec:alg_lang}
\input{tex/alg_lang}

\section{Target Language}
\label{sec:target_lang}
\input{tex/target_lang}

\section{Lowering}
\label{sec:lowering}
\input{tex/lowering}

\section{Bounds Inference}
\label{sec:bounds}
\input{tex/bounds}

\section{Scheduling Language}
\label{sec:scheduling}
\input{tex/scheduling}

\section{Related Work}
\label{sec:related_work}
\input{tex/related_work}

\section{Practical impact}
% \balance
\label{sec:conclusion-impact}
\input{tex/impact}

\begin{acks}
\input{tex/acknowledgments}
\end{acks}

\bibliography{main}

\clearpage
\appendix
\input{tex/appendices}

\end{document}

%% file: tex/prelude.tex
\makeatletter
\def\mdseries@tt{m}
\makeatother

\usepackage{amsbsy}
\usepackage{amscd}
\usepackage{amsfonts}
\usepackage{amsmath}
\usepackage{amsthm}
\usepackage{booktabs}  % rec'd by POPL
\usepackage{fancybox}
\usepackage{latexsym}
\usepackage{microtype}
\usepackage{multirow}
\usepackage{stmaryrd}
\usepackage{tabularx}
\usepackage[labelformat=simple]{subcaption}

\usepackage{bussproofs} % big problem -- this is a veerrrryyy useful package.
\usepackage{enumitem}   % no idea why this isn't allowed
\usepackage{framed}     % use fancybox instead?
\usepackage{mathtools}  % for DeclarePairedDelimiter (ceil) but really why isn't this approved?
\usepackage{mdframed}   % for more configurable frames
\usepackage{tikz-cd}    % for box and arrow diagrams
\usepackage{wrapfig}

\usepackage[disabled]{aplcomments} % Should not be included in the final paper. No need to replace.

\usetikzlibrary{shapes,arrows}

\definecolor{purple}{RGB}{153,127,186}
\definecolor{orange}{RGB}{248,154,36}
\colorlet{darkpurple}{purple!75!black}
\colorlet{darkorange}{orange!75!black}
\colorlet{lightpurple}{purple!25}
\colorlet{lightorange}{orange!25}

\newcommand{\holef}[1]{\begingroup\setlength{\fboxsep}{0pt}\colorbox{orange}{\strut\ensuremath{~#1~}}\endgroup}
\newcommand{\holeg}[1]{\begingroup\color{white}\setlength{\fboxsep}{0pt}\colorbox{purple}{\strut\ensuremath{~#1~}}\endgroup}

\newcommand{\xlo}{x_{\min}}
\newcommand{\xlen}{x_{\len}}

\tikzset{%
  point_g/.style = {draw = darkpurple, fill = purple, thick, rectangle, minimum height = 2.5em, minimum width = 2.5em, text = white},
  point_f/.style = {draw = darkorange, fill = orange, thick, rectangle, minimum height = 2.5em, minimum width = 2.5em},
  point_g_gb/.style = {draw = darkpurple, fill = lightpurple, thick, dashed, rectangle, minimum height = 2.5em, minimum width = 2.5em},
  point_f_gb/.style = {draw = darkorange, fill = lightorange, thick, dashed, rectangle, minimum height = 2.5em, minimum width = 2.5em},
}

\setlist{leftmargin=1em}

\renewcommand{\implies}{\DOTSB\Rightarrow}

\newcommenter{JRK}{0.0,0.0,1.0}
\newcommenter{Alex}{0.0,1.0,0.0}
\newcommenter{glb}{1.0,0.0,0.0}

\newif\ifproofs
\proofsfalse

\newcommand{\bigeval}{\Downarrow}
\newcommand{\cpu}{\mathrm{cpu}}
\newcommand{\mem}{\mathrm{mem}}

\begingroup
\lccode`\~=`\|
\lowercase{\endgroup
\newcommand{\mkbrak}[3]{{\def~{\;\middle\vert\;}\mathcode`\|="8000 \ensuremath{\left#1 #3 \right#2} }}}
\newcommand{\tup}[1]{\mkbrak{\langle}{\rangle}{#1}}

\newcommand{\set}[1]{\mkbrak{\{}{\}}{#1}}

\DeclareMathOperator{\Div}{div}
\DeclareMathOperator{\Mod}{mod}
\DeclareMathOperator{\len}{len}
\DeclareMathOperator{\select}{select}
\DeclarePairedDelimiter{\ceil}{\lceil}{\rceil}
\DeclarePairedDelimiter{\floor}{\lfloor}{\rfloor}

\newcommand{\kwd}[1]{\operatorname{\mathbf{#1}}}

\newcommand{\svec}[1]{\overline{#1}}

\newsavebox{\valignbox}

\acmJournal{PACMPL}
\acmVolume{1}
\acmNumber{OOPSLA} % CONF = POPL or ICFP or OOPSLA
\acmArticle{1}
\acmYear{2022}
\acmMonth{1}
\acmDOI{} % \acmDOI{10.1145/nnnnnnn.nnnnnnn}
\startPage{1}

\setcopyright{none}
\author{Alex Reinking}
\email{alex_reinking@berkeley.edu}
\affiliation{%
  \institution{University of California, Berkeley}
  \country{United States}}

\author{Gilbert Louis Bernstein}
\email{gilbo@berkeley.edu}
\affiliation{%
  \institution{University of California, Berkeley}
  \country{United States}}

\author{Jonathan Ragan-Kelley}
\email{jrk@mit.edu}
\affiliation{%
  \institution{Massachusetts Institute of Technology}
  \country{United States}}

\ccsdesc[500]{Software and its engineering~Semantics}
\ccsdesc[500]{Software and its engineering~Domain specific languages}
\ccsdesc[300]{Software and its engineering~Functionality}
\ccsdesc[300]{Software and its engineering~Consistency}

%% file: tex/abstract.tex
\begin{abstract}
We present the first formalization and metatheory of language soundness for a user-schedulable language, the widely used array processing language Halide.
User-schedulable languages strike a balance between abstraction and control in high-performance computing by separating the specification of \emph{what} a program should compute from a schedule for \emph{how} to compute it.
In the process, they make a novel language soundness claim: the result of a program should always be the same, regardless of how it is scheduled.
This soundness guarantee is tricky to provide in the presence of schedules that introduce redundant recomputation and computation on uninitialized data, rather than simply reordering statements.
In addition, Halide ensures memory safety through a compile-time \emph{bounds inference} engine that determines safe sizes for every buffer and loop in the generated code, presenting a novel challenge: formalizing and analyzing a language specification that depends on the results of unreliable program synthesis algorithms.
Our formalization has revealed flaws and led to improvements in the practical Halide system, and we believe it provides a foundation for the design of new languages and tools that apply programmer-controlled scheduling to other domains.
\end{abstract}

%% file: tex/introduction.tex
Halide is a domain-specific language used widely in industry to build high-performance image and array processing pipelines for everything from YouTube, to every Android phone, to Adobe Photoshop~\cite{Halide:SIGGRAPH,Halide:CACM,HDRPlus,PixelVisualCore}.
As part of a new generation of \emph{user-schedulable} languages and compilers~\cite{TVM,TACO,GraphIt,Tiramisu,Taichi,TeML,Sequoia,FireIron,PolyMage},
its design separates the specification of what is to be computed, known as the \emph{algorithm}, from the specification of when and where those computations should be carried out and placed in memory, known as the \emph{schedule}, while allowing both to be supplied by the programmer.

The key value of user scheduling is that programmers are relieved from troubleshooting large classes of bugs that arise when optimizing programs for memory locality, parallelism, and vectorization because these transformations are available in the scheduling language, rather than being directly expressed in the algorithm.
This enables a schedule-centric workflow where the majority of effort is spent exploring different optimizations, not ensuring correctness after each attempt.
This allows performance engineers to optimize programs competitively with the best hand-tuned C, assembly, and CUDA implementations, but with dramatically less code and development time~\cite{Halide:CACM}.

Halide specifies algorithms in a purely functional dataflow language of infinite arrays that combines lazy and eager semantics.
The schedule then guides compilation to generate some particular eager, imperative implementation.
Halide schedules include classic loop \emph{re-ordering} transformations, but their most unique constructs ($\kwd{compute-at}$, $\kwd{store-at}$) induce non-local transformations that intentionally exploit redundant recomputation and computation on uninitialized data---transformations well outside that classic model.

So changing the schedule of a Halide program dramatically changes the computation, but when is this safe and sound?

The safety and soundness of languages with user-controlled scheduling has never been formally defined and analyzed, particularly in the presence of non-reordering transformations (\S\ref{sec:related_work}).
This paper presents the first formal definition and analysis of the core of Halide, and a general approach to the metatheory of similar languages.
We focus on proving a new safety and correctness guarantee unique to user-scheduled languages: regardless of what schedule is supplied, a given algorithm should always produce the same result and be memory-safe.

Formalizing the correctness of Halide is difficult for at least two reasons.
First, traditional formalisms for reasoning about the correctness of compiler transformations (especially loop transformations) tend to reduce to dependence analysis of imperative code.
This strategy is only applicable to re-ordering transformations, and so we must build our proofs and reasoning on a different structural basis (\S\ref{sec:overview}).
Second, Halide's design is built around a \emph{bounds inference} engine that assumes responsibility for synthesizing all loop and memory-buffer bounds.
Because the synthesis of finite bounds is undecidable in the general case with data-dependent accesses and non-affine expressions, compliant bounds inference engines must be allowed to fail. This opens the door to compliant but useless engines (which always fail).
We propose a solution to this conundrum, by specifying a reference algorithm to define minimum compliant quality (\S\ref{sec:overview-bi}).

Halide's success in industry has, for better or worse, locked it in to its early design decisions and has influenced the design of its peers.
At the same time, these systems have not historically been a subject of interest in formal programming languages.
An important consequence of our work has been to crisply define Halide's semantics and metatheory and to correct mistakes without dramatically overhauling the language (\S\ref{sec:conclusion-impact}).
We further expect this effort can help the next wave of user-schedulable languages to create even more elegant and useful systems without making the same compromises. % too soft?

This paper makes the following contributions:

\begin{itemize}
    \item We give the first complete semantics and metatheory defining \emph{sound} user-specified scheduling of a high-performance array processing language: that programs are unconditionally memory safe and that their output is not changed by scheduling decisions.
    \item We give the first precise description of the core of the practical Halide system: the algorithm language, scheduling operators, and bounds inference problem.
    \item We provide the first definition of Halide's bounds inference feature as a program synthesis problem, which was not previously understood as such.
    \item We apply our formalism to the practical Halide system, finding \& fixing several bugs and making design improvements in the process.
\end{itemize}

%% file: tex/example.tex
To introduce key concepts and build intuition for the formalism to follow, we'll consider a minimal example that showcases the challenges present in analyzing the scheduling language.
Our example algorithm consists of two \emph{``funcs''}, defined like so:
\begin{align*}
& \hspace{0em} \kwd{pipeline} f(): \\
& \hspace{1em}   \kwd{fun} g(x) = \ldots; \\
& \hspace{1em}   \kwd{fun} f(x) = g(x) + g(x+1);
\end{align*}
A \emph{func} is the basic unit of computation.
Halide funcs are defined on unbounded, $n$-dimensional, integer lattices, and are \emph{not} bounded, multidimensional arrays.
The body of $g$ is deliberately left undefined since it is irrelevant to the upcoming discussion.
The formula describing a func must be \emph{total} and \emph{non-recursive}, so that any window of the func has defined values.
To run the pipeline, the user supplies as \emph{input}%
\footnote{For simplicity of formalization, we omit special treatment of \emph{input} funcs, modeling these instead as procedural funcs with no dependencies.
However, the practical system does support input arrays (whose bounds must be checked for consistency when the program starts running).}
a desired window over which to compute the last func in the pipeline.
The program then returns an array (formalized as partial functions) containing the computed values and which might be larger than requested.
The computational model is therefore \emph{demand-driven}, unlike most contemporary array languages.
We will illustrate the evaluation of $f$ on the window $[0,6)$, which will in turn necessitate computing $g$ on at least the window $[0,7)$.

\input{figures/example/basic}
\input{figures/example/recompute}
\input{figures/example/overcompute}

In order to recover an imperative implementation, we \emph{lower} the pipeline into a second, imperative, target language with C-like semantics.
While there exist sensible choices for loop iteration bounds and buffer sizes, our algorithm never specified these.
Therefore, this initial lowered program leaves symbolic \emph{holes} in the code (prefixed with `$?$').

To fill these holes, Halide performs \emph{bounds inference}, which we formalize as a program synthesis problem.
We assume a bounds inference oracle that returns expressions to fill every hole, satisfying derived memory safety and correctness conditions.
However, since this oracle is only required to meet safety and correctness conditions, there is no guarantee on the (parameterized) minimality of memory allocations or loop bounds.
We will discuss this complication in more detail shortly.

Now we will look at three different ways to schedule our pipeline.

First, the default schedule (Figure \ref{fig:example-basic}) computes all values required from each func before progressing to the next func, in order of their definition.
Note that $x_\textrm{min}$ and $x_\textrm{len}$ are variables specifying the output window $[x_\textrm{min}, x_\textrm{min}+x_\textrm{len})$.
Bounds inference could efficiently fill the hole $?^\textrm{mem} f_x$ with $[x_\textrm{min}, x_\textrm{min}+x_\textrm{len})$ and the holes $?^\textrm{cpu}g_x$ and $?^\textrm{mem}g_x$ with $[x_\textrm{min}, x_\textrm{min}+x_\textrm{len}+1)$.
An inefficient solution could fill $?^\textrm{mem}g_x$ with $[x_\textrm{min}, x_\textrm{min}+2\cdot x_\textrm{len})$, but $[x_\textrm{min}, x_\textrm{min}+x_\textrm{len}-1)$ is not allowed because the last access would write out of bounds.

For our second schedule of $f$ (Figure \ref{fig:example-recompute}), we will tile it so that we can parallelize it, computing $f$ in independent $3$-element-wide tiles.
This first scheduling directive says to $\kwd{split}(\tup{f, x}, xo, xi, 3)$ the computation of $f$ along dimension $x$ by a factor of $3$ into an outer iteration dimension $xo$ and inner iteration dimension $xi$.
Then, the second directive tells us $\kwd{compute-at}(g, \tup{f,xo})$, meaning to re-compute the necessary portion of $g$ at $f$, within iteration level $xo$, and then to $\kwd{store-at}(g, \tup{f,xo})$, similarly.
In terms of imperative code, this is simply a relocation of the loop nest computing $g$.
Bounds inference will now be able to infer much \emph{tighter} bounds on $g$, since it only needs to be computed on a per-tile basis.
For a given value of $xo$, only $4$ values of $g$ need to be computed and stored for use by the $xi$ loop.

Notice that the windows of $g$ required by adjacent tiles of $f$ overlap by one element.
In Figure \ref{fig:example-recompute-diagram} we can see that $g(3)$ is required by both tiles of $f$ because both $f(2)$ and $f(3)$ depend on it.
This ability to reduce synchronization and improve locality at the expense of redundant recomputation is at the heart of why Halide is able to generate high-performance code for modern micro-processors.
It is also one reason why only using re-ordering loop transformations is insufficient.

For our third and final schedule (Figure \ref{fig:example-overcompute}), we will tile the computation of $f$ in order to take advantage of fixed-width SIMD instructions (present on most CPUs today).
To do so, we call $\kwd{split}$ again, but now provide an alternate \emph{tail-strategy}: $\varphi_\text{Round}$.
Rather than introducing an if-guard, this strategy will cause $f$ to be \emph{unconditionally} evaluated in 4-wide tiles.
If the requested window is not a multiple of 4, it will be rounded up and extra points will be computed.

Perhaps counter-intuitively, this \emph{over-compute} strategy requires \emph{fewer} instructions in a vectorized implementation, since the entire loop tail can be computed with a single instruction, rather than a variable number of scalar operations (e.g. in a loop epilogue).
However, whereas the window of allocated, computed, and valid values all coincided before, those 3 windows now all uncouple.
For the requested window $[0,6)$, $f$ is allocated and computed on the window $[0,8)$, whereas $g$ is allocated on the window $[0,9)$ and computed on the window $[0,7)$.
Since neither $g(7)$ nor $g(8)$ are initialized, the values in $f(6)$ and $f(7)$ are themselves uninitialized.

The IR programs in Figures \ref{fig:example-basic-ir}, \ref{fig:example-recompute-ir}, and \ref{fig:example-overcompute-ir} show an important benefit of user-specified scheduling: in a traditional high performance language like C, a programmer would need to write loops and derive compute bounds by hand.
By instead factoring these rewrites into a small scheduling language, Halide programmers can efficiently explore the space of safe, equivalent programs.
In this paper, we explain how---formally---this promise to Halide programmers is justified.

%% file: figures/example/basic.tex
\begin{figure*}[t]
\begin{subfigure}[b]{0.49\textwidth}
\centering
\small
\begin{tikzpicture}[remember picture]
\node(1S){\emph{(default schedule)}} ;
\end{tikzpicture} \\
\begin{tikzpicture}[auto, thick, node distance=2.5em, >=latex]
    % Draw the grid of function points
    % Func G
    \draw
    	node at (0,0) [point_g] (g0) {$g(0)$}
    	node [point_g, right of=g0] (g1) {$g(1)$}
    	node [point_g, right of=g1] (g2) {$g(2)$}
    	node [point_g, right of=g2] (g3) {$g(3)$}
    	node [point_g, right of=g3] (g4) {$g(4)$}
    	node [point_g, right of=g4] (g5) {$g(5)$}
    	node [point_g, right of=g5] (g6) {$g(6)$}
    ;
    % Func F
    \draw
    	node at (0,-5em) [point_f] (f0) {$f(0)$}
    	node [point_f, right of=f0] (f1) {$f(1)$}
    	node [point_f, right of=f1] (f2) {$f(2)$}
    	node [point_f, right of=f2] (f3) {$f(3)$}
    	node [point_f, right of=f3] (f4) {$f(4)$}
    	node [point_f, right of=f4] (f5) {$f(5)$}
    ;
    % Dependency arrows
    \draw[->](f0) -- node {} (g0);
    \draw[->](f0) -- node {} (g1);
    \draw[->, draw=black!50](f1) -- node {} (g1);
    \draw[->, draw=black!50](f1) -- node {} (g2);
    % Computation
    \draw[->, draw=red!80!black]
      ([yshift=1em]g0.base) -- ([yshift=1em]g6.base);
    \draw[->, draw=red!80!black, dotted]
      ([yshift=1em]g6.base) -- ([yshift=1em]f0.base);
    \draw[->, draw=red!80!black]
      ([yshift=1em]f0.base) -- ([yshift=1em]f5.base);
\end{tikzpicture}
\caption{Points of $g$ are shown in purple, and points of $f$ are shown in orange. Dark arrows indicate dependencies; for clarity, only those for $f(0)$ and $f(1)$ are shown. Red arrows trace computation; all points of $g$ are computed first, and only then is $f$ computed.}
\label{fig:example-basic-diagram}
\end{subfigure}\hfill%
\begin{subfigure}[b]{0.49\textwidth}
\centering
\footnotesize
\begin{tikzpicture}[remember picture, auto, thick, node distance=2.5em, >=latex]
  \node(1IR){$
  \begin{aligned}
  & \kwd{pipeline} f(\xlo,\xlen): \\
  & \hspace{1em} \kwd{allocate} g(\holeg{?^{\mem}g_x}); \\
  & \hspace{1em} \kwd{label} g :
    \hspace{0em}   \kwd{label} s_0 : \\
  & \hspace{2em}     \kwd{for} x \kwd{in} ~\holeg{?^{\cpu}g_x} \kwd{do} \\
  & \hspace{3em}        g[x] \leftarrow \ldots; \\
  & \hspace{1em} \kwd{allocate} f(\holef{?^{\mem}f_x}); \\
  & \hspace{1em} \kwd{label} f :
    \hspace{0em}   \kwd{label} s_0 : \\
  & \hspace{2em}     \kwd{for} x \kwd{in} ~\holef{(\xlo, \xlen)} \kwd{do} \\
  & \hspace{3em}       f[x] \leftarrow g[x] + g[x+1];
  \end{aligned}$};
\end{tikzpicture}
\begin{tikzpicture}[remember picture, overlay, >=latex]
\path (1S) edge[->,out=5,in=150] node[midway,fill=white] {\emph{results in}} (1IR);
\end{tikzpicture}
\caption{Loop nest IR with default schedule. Note the \emph{holes} that will be filled in by bounds inference.}
\label{fig:example-basic-ir}
\end{subfigure}
\caption{Example program with default execution order.}
\label{fig:example-basic}
\end{figure*}
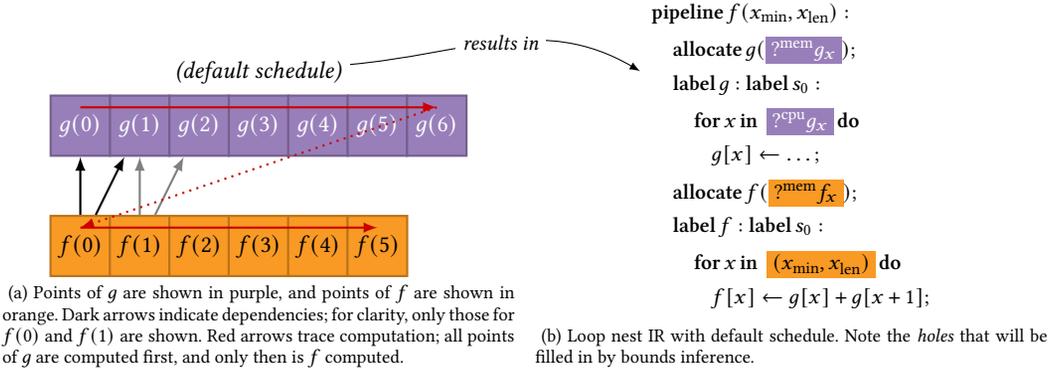

%% file: figures/example/recompute.tex
\begin{figure*}[t]
\begin{subfigure}[b]{0.49\textwidth}
\centering
\footnotesize
\begin{tikzpicture}[remember picture]
\node(2S){$
\begin{aligned}
& \kwd{split}(\tup{f,x},xo,xi,3) ;\\
& \kwd{compute-at}(g, \tup{f,xo}) ;\\
& \kwd{store-at}(g, \tup{f,xo}) ;
\end{aligned}
$};
\end{tikzpicture} \\
\begin{tikzpicture}[auto, thick, node distance=2.5em, >=latex]
    % Draw the grid of function points
    % Func G
    \draw
    	node at (0,0) [point_g] (g0) {$g(0)$}
    	node [point_g, right of=g0] (g1) {$g(1)$}
    	node [point_g, right of=g1] (g2) {$g(2)$}
    	node [point_g, right of=g2] (g3) {$g(3)$}
    	node [point_g, right of=g3] (g4) {$g(4)$}
    	node [point_g, right of=g4] (g5) {$g(5)$}
    	node [point_g, right of=g5] (g6) {$g(6)$}
    ;
    % Func F
    \draw
    	node at (0,-5em) [point_f] (f0) {$f(0)$}
    	node [point_f, right of=f0] (f1) {$f(1)$}
    	node [point_f, right of=f1] (f2) {$f(2)$}
    	node [point_f, right of=f2] (f3) {$f(3)$}
    	node [point_f, right of=f3] (f4) {$f(4)$}
    	node [point_f, right of=f4] (f5) {$f(5)$}
    ;
    %% Split
    \draw[very thick, black, dashed] ([xshift=-0.5pt, yshift=-0.5em]f2.south east) -- ([xshift=-0.5pt, yshift=0.5em]f2.north east);
    %% Deps
    \draw[->, draw=black](f2) -- node {} (g3);
    \draw[->, draw=black](f3) -- node {} (g3);
    %% Computation
    % Tile 1
    \draw[->, draw=red!80!black]
      ([yshift=1em]g0.base) -- ([yshift=1em]g3.base);
    \draw[->, draw=red!80!black, dotted]
      ([yshift=1em]g3.base) -- ([yshift=1em]f0.base);
    \draw[->, draw=red!80!black]
      ([yshift=1em]f0.base) -- ([yshift=1em]f2.base);
    % Cross-tile
    \draw[->, draw=red!80!black, dotted]
      ([yshift=1em]f2.base) to [out=90,in=180] ([yshift=-0.5em]g3.base);
    % Tile 2
    \draw[->, draw=red!80!black]
      ([yshift=-0.5em]g3.base) -- ([yshift=-0.5em]g6.base);
    \draw[->, draw=red!80!black, dotted]
      ([yshift=-0.5em]g6.base) -- ([yshift=1em]f3.base);
    \draw[->, draw=red!80!black]
      ([yshift=1em]f3.base) -- ([yshift=1em]f5.base);
\end{tikzpicture}
\caption{This time the computation of $f$ is split by 3, and $g$ is \emph{interleaved} per-tile. Bounds inference causes $g(3)$ to be \emph{redundantly recomputed}. This allows each tile of $f$ to be trivially parallelized.}
\label{fig:example-recompute-diagram}
\end{subfigure}\hfill%
\begin{subfigure}[b]{0.49\textwidth}
\centering
\footnotesize
\begin{tikzpicture}[remember picture]
\node(2IR){$
\begin{aligned}
  & \kwd{pipeline} f(\xlo,\xlen): \\
  & \hspace{1em} \kwd{allocate} f(\holef{?^{\mem}f_x}); \\
  & \hspace{1em} \kwd{label} f :
    \hspace{0em}   \kwd{label} s_0 : \\
  & \hspace{2em}     \kwd{for} xo \kwd{in} ~\holef{(\xlo, (\xlen + 2)/3)} \kwd{do} \\
  & \hspace{3em}       \holeg{\text{compute $g$ as needed}} \\
  & \hspace{3em}       \kwd{for} xi \kwd{in} ~\holef{[0, 3)} \kwd{do} \\
  & \hspace{4em}         \kwd{let} x = \xlo + 3 xo + xi ~\kwd{in} \\
  & \hspace{4em}         \kwd{if} x < \xlo + \xlen \kwd{then} \\
  & \hspace{5em}           f[x] \leftarrow g[x] + g[x+1];
\end{aligned}
$};
\end{tikzpicture}
\begin{tikzpicture}[remember picture, overlay, >=latex]
\path (2S) edge[->,out=15,in=150] node[midway,fill=white] {\emph{results in}} (2IR);
\end{tikzpicture}
\caption{Loop nest IR with recomputing schedule. The loops for $g$ have been abbreviated for space, but are identical to those in Figure \ref{fig:example-basic-ir}.}
\label{fig:example-recompute-ir}
\end{subfigure}
\caption{Example program after tiling $f$ by 3. This schedule illustrates Halide's ability to introduce \emph{redundant recomputation} to achieve better producer-consumer locality.}
\label{fig:example-recompute}
\end{figure*}
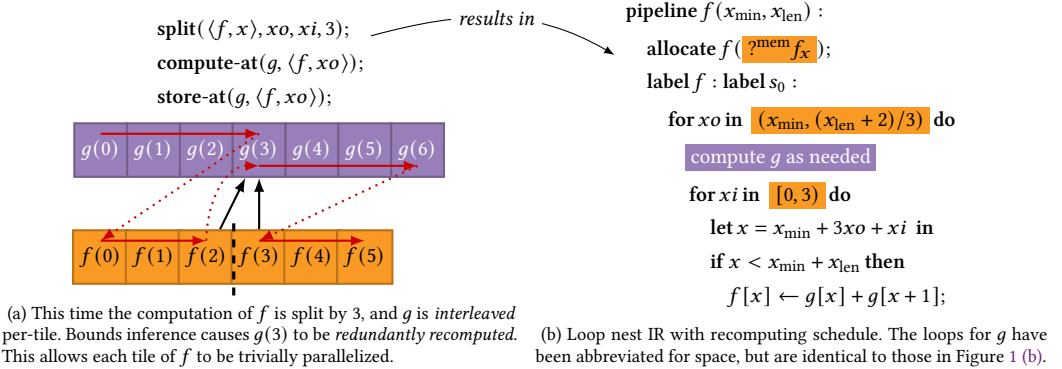

%% file: figures/example/overcompute.tex
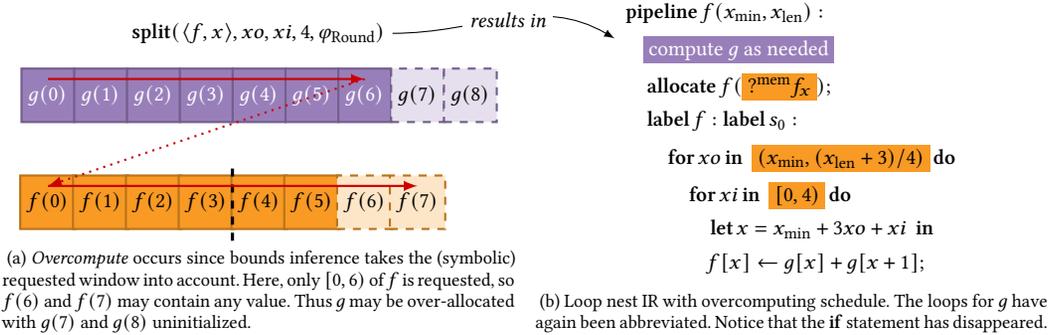
\begin{figure*}[t]
\begin{subfigure}[b]{0.49\textwidth}
\centering
\footnotesize
\begin{tikzpicture}[remember picture]
\node[inner ysep=1em](3S){$\kwd{split}(\tup{f,x},xo,xi,4,\varphi_\text{Round})$};
\end{tikzpicture}
\begin{tikzpicture}[auto, thick, node distance=2.5em, >=latex]
    % Draw the grid of function points
    % Func G
    \draw
    	node at (0,0) [point_g] (g0) {$g(0)$}
    	node [point_g, right of=g0] (g1) {$g(1)$}
    	node [point_g, right of=g1] (g2) {$g(2)$}
    	node [point_g, right of=g2] (g3) {$g(3)$}
    	node [point_g, right of=g3] (g4) {$g(4)$}
    	node [point_g, right of=g4] (g5) {$g(5)$}
    	node [point_g, right of=g5] (g6) {$g(6)$}
    	node [point_g_gb, right of=g6] (g7) {$g(7)$}
    	node [point_g_gb, right of=g7] (g8) {$g(8)$}
    ;
    % Func F
    \draw
    	node at (0,-5em) [point_f] (f0) {$f(0)$}
    	node [point_f, right of=f0] (f1) {$f(1)$}
    	node [point_f, right of=f1] (f2) {$f(2)$}
    	node [point_f, right of=f2] (f3) {$f(3)$}
    	node [point_f, right of=f3] (f4) {$f(4)$}
    	node [point_f, right of=f4] (f5) {$f(5)$}
    	node [point_f_gb, right of=f5] (f6) {$f(6)$}
    	node [point_f_gb, right of=f6] (f7) {$f(7)$}
    ;
    %% Split
    \draw[very thick, black, dashed] ([xshift=-1pt, yshift=-0.5em]f3.south east) -- ([xshift=-1pt, yshift=0.5em]f3.north east);
    %% Computation
    \draw[->, draw=red!80!black]
      ([yshift=1em]g0.base) -- ([yshift=1em]g6.base);
    \draw[->, draw=red!80!black, dotted]
      ([yshift=1em]g6.base) -- ([yshift=1em]f0.base);
    \draw[->, draw=red!80!black]
      ([yshift=1em]f0.base) -- ([yshift=1em]f7.base);
\end{tikzpicture}
\caption{\emph{Overcompute} occurs since bounds inference takes the (symbolic) requested window into account. Here, only $[0,6)$ of $f$ is requested, so $f(6)$ and $f(7)$ may contain any value. Thus $g$ may be over-allocated with $g(7)$ and $g(8)$ uninitialized.}
\label{fig:example-overcompute-diagram}
\end{subfigure}\hfill%
\begin{subfigure}[b]{0.49\textwidth}
\centering
\footnotesize
\begin{tikzpicture}[remember picture]
\node(3IR){$
\begin{aligned}
  & \kwd{pipeline} f(\xlo,\xlen): \\
  & \hspace{1em} \holeg{\text{compute $g$ as needed}} \\
  & \hspace{1em} \kwd{allocate} f(\holef{?^{\mem}f_x}); \\
  & \hspace{1em} \kwd{label} f :
    \hspace{0em}   \kwd{label} s_0 : \\
  & \hspace{2em}     \kwd{for} xo \kwd{in} ~\holef{(\xlo, (\xlen + 3)/4)} \kwd{do} \\
  & \hspace{3em}       \kwd{for} xi \kwd{in} ~\holef{[0, 4)} \kwd{do} \\
  & \hspace{4em}         \kwd{let} x = \xlo + 3 xo + xi ~\kwd{in} \\
  & \hspace{4em}         f[x] \leftarrow g[x] + g[x+1];
\end{aligned}
$};
\end{tikzpicture}
\begin{tikzpicture}[remember picture, overlay, >=latex]
\path (3S) edge[->,out=0,in=150] node[midway,fill=white] {\emph{results in}} (3IR);
\end{tikzpicture}
\caption{Loop nest IR with overcomputing schedule. The loops for $g$ have again been abbreviated. Notice that the $\kwd{if}$ statement has disappeared.}
\label{fig:example-overcompute-ir}
\end{subfigure}
\caption{Example program after vectorizing $f$ and rounding up. This schedule illustrates Halide's ability to introduce \emph{overcompute} on \emph{uninitialized values} to trade-off between compute and storage efficiency.}
\label{fig:example-overcompute}
\end{figure*}

%% file: tex/overview.tex
Given an initial program $P_0$, and a schedule $S= s_1;\hdots; s_n$, let $P_i$ be the result of applying the scheduling primitive $s_i$ to $P_{i-1}$.
Intuitively, if we can show that our soundness property is \emph{invariant} under each possible primitive, then it must hold.

In loop-nest optimization this invariant was traditionally specified via dependence graphs: first over lexical statements, and then over whole iteration spaces of \emph{statement instances}.
For instance, \citet{Kennedy:textbook} formulate this invariant as the \emph{Fundamental Theorem of Dependence}, which states ``any reordering transformation that preserves every dependence in a program preserves the meaning of that program.''
Unfortunately, this approach is strictly limited to verifying the soundness of \emph{reordering transformations}, i.e. those transformations that permute the order in which statement instances occur, but never change those statements, duplicate them, or introduce new ones.

We resolve this problem in our proof structure by including the original \emph{provenance} of the programs $P_i$ in our sound\-ness invariant.
In Halide, this original reference program is the \emph{algorithm}, which is expressed in a functional, rather than imperative, language.
Since the functional algorithm does not specify order of execution, nor where and how values are stored in memory buffers, this soundness principle accommodates a greater range of transformations.

\input{figures/overview/system}
Finally, loop transformations on array code inevitably require complicated reasoning about various sets of bounds (e.g., for memory allocation) as transformations are performed.
In Halide, these complexities are managed by deferring bounds analyses until \emph{after} scheduling is performed, and by offloading those decisions to a \emph{bounds inference} engine, which we treat here as an oracle.
More generally, we expect that advances in program synthesis will only make the transformation of incomplete programs more common; this general approach should work in a variety of new language designs.

In order to handle the transformation of programs with holes, our soundness invariant must be stated on sets of programs (completions) rather than individual programs.
Thus, rather than state that each transformed program is consistent with the algorithm reference, we require that all completions are consistent---where defined.

\subsection{Basic Definitions}
Halide programs are specified via an algorithm program denoted $P \in\kwd{Alg}$, written in a functional language with big-step semantics, defined in \S\ref{sec:alg_lang}.
We immediately \emph{lower} this algorithm into an imperative target language program \emph{with holes} denoted $T\in\kwd{Tgt}^?$.
This language is defined in \S\ref{sec:target_lang}.
Lowering is specified via a function $T = \mathcal{L}(P)$, defined in \S\ref{sec:lowering}.

The lowered program is incomplete because it is missing various bounds.
Halide's bounds inference completes a program with holes $T\in\kwd{Tgt}^?$ into a program without holes $P' \in\kwd{Tgt}$.
We model the bounds-inference procedure as a non-deterministic oracle $\kwd{BI}$, which defines a set of completions $\kwd{BI}(T)$ via a syntactically derived synthesis problem, and returns one as specified in \S\ref{sec:bounds}.
Thus, $P' \in \kwd{BI}(T) \subseteq \kwd{Tgt}$.

The schedule for a Halide program is specified as a sequence of primitive scheduling directives $S= s_1;\dots; s_n$, defined in \S\ref{sec:scheduling}.
Scheduling proceeds by sequentially transforming a target program with holes $T_0 = \mathcal{L}(P)$ by each subsequent scheduling directive such that $T_{i+1} = \mathcal{S}(s_i, T_i)$.
Consequently, a set of completions $\kwd{BI}(T_i)$ is defined at each point in the scheduling process.
These intermediate completions are not used when simply compiling a program, but are essential to analyzing the behavior of a scheduling directive by relating the sets before and after the transformation.
The fully scheduled program is given as $P'_n \in \kwd{BI}(\mathcal{S}(S,\mathcal{L}(P)))$.
This structure is depicted in Figure~\ref{fig:overview-system}.

\subsection{Equivalence and Soundness of Programs}
Halide makes two fundamental promises to programmers: memory safety and equivalence under scheduling transformations.
Here is how we formulate those promises.

Programs are defined as functions of \emph{input parameters} and an \emph{output window}.
Thus, the same program can be run multiple times to compute different windows into a conceptually unbounded output array.

\begin{definition}[Input and Output]\label{def:meta-inputs}
Let $P\in\kwd{Alg}$ or $P\in\kwd{Tgt}$ be a program with $m$ parameters and an $n$ dimensional output func.
An \emph{input} $z$ to $P$ is an assignment to those $m$ parameters and an assignment of $n$ constant intervals defining an output window $R(z)$.
The \emph{output} of running $P$ on $z$ is a partial function $f = P(z)$ where $f(x)$ is defined on at least all $x\in R(z)$.
\end{definition}

For a variety of reasons, a program may produce \emph{more} than the requested output window $R(z)$.
A program may even allocate padding space and fill it with garbage values in order to align storage and/or computation.
For these reasons, we only define equivalence up to agreement on the specified output window:

\begin{definition}[Output equivalence]\label{def:output-equiv}
Let each of $P$ and $P'$ be either an algorithm or target language program, with common input $z$. We say they have equivalent outputs $P \simeq_z P'$, if for every point $x \in R(z)$, $P(z)(x)=P'(z)(x)$.
\end{definition}

This definition of equivalence is sufficient to compare two complete programs.
However, because our soundness invariant must be stated on incomplete programs $T \in \kwd{Tgt}^?$, we will define the \emph{confluence} of an algorithm with all completions of $T$.
We will also have to account for certain exceptional cases in which the output may actually not be equivalent.
Namely, if the original algorithm contains errors, then all bets are off, and if the completion of the program fails to satisfy bounds-constraints, then equivalence cannot be guaranteed.
This latter case should be concerning; we will address it shortly.

\begin{definition}[Algorithm confluence]\label{def:confluence}
Let $P\in\kwd{Alg}$ be a Halide algorithm and let $T\in\kwd{Tgt}^?$ be a target language program with holes.
We say that $T$ is \emph{confluent} with $P$ if for all $P' \in\kwd{BI}(T)$ and all inputs $z$, either $P(z)$ contains an \emph{error value} in $R(z)$, $P'(z)$ \emph{fails an assertion check}, or $P \simeq_z P'$.
Error values and assertion failures are detailed in \S\ref{sec:alg_lang-expressions} and \S\ref{sec:target_lang-semantics}, respectively.
\end{definition}

We are now able to state the two fundamental theorems about Halide.
In stating these theorems, we assume that algorithm language programs are \emph{valid}~(\S\ref{sec:alg_lang-rules}), as are schedules~(\S\ref{sec:scheduling}).

\begin{theorem}[Memory safety]\label{thm:meta-safety}
Let $P\in\kwd{Alg}$ be a valid program, $z$ an input, and $S$ a valid schedule.
Then, for all target language programs $P' \in \kwd{BI}(\mathcal{S}(S,\mathcal{L}(P)))$, the computation $P'(z)$ will not access any out of bounds memory (\S\ref{sec:target_lang-semantics}).
\end{theorem}

Memory safety will be guaranteed by the bounds inference oracle.
The problem posed to this oracle is defined in \S\ref{sec:bounds} such that safety is provided by construction.

\begin{theorem}[Scheduling equivalence]\label{thm:meta-equiv}
Let $P\in\kwd{Alg}$ be a valid program and $S$ any valid schedule.
Then all target language programs $P' \in \kwd{BI}(\mathcal{S}(S,\mathcal{L}(P)))$ are confluent with $P$.
\end{theorem}

The preceding property constitutes our soundness invariant.
Proving the theorem therefore reduces to showing that this invariant is preserved first by lowering, and then by each subsequent possible primitive scheduling transformation.

\subsection{Bounds Inference and Language Specification}
\label{sec:overview-bi}

The definition of algorithm confluence permits the bounds inference oracle to insert assertion checks that might fail into completed programs.
This design presents a unique challenge: a completion that \emph{always} fails an assertion check is technically confluent with its original algorithm, but is not useful.
Less vacuously, as the bounds inference engine improves, the set of scheduled programs for which we find good---or even just satisfactory---bounds changes.
Hopefully, the result is a strict improvement but regressions are possible and even likely in the compiler.

There is a strong case that Halide's design is wrong because there are no guarantees that any given program will continue to work without assertion failures when run on different versions of the standard compiler, much less on alternative implementations.
At a minimum, such a fact runs counter to the spirit of specifying programming language behavior.
So, we might be tempted to try to re-design Halide.

Instead, we choose to tackle specifying Halide's existing design for two main reasons.
First, Halide has been in industrial use for nearly a decade, shipped in many products, and would benefit more from specification of its actual design than of an idealization.
Second, flexibility around bounds inference in Halide is essential to array processing and the ability to reason about redundant re-computation and over-computation.
An alternative to bounds inference would be intriguing, but also constitute a novel advance in language design on its own.

Our strategy is to supply a baseline, or a lower bound, for the quality of the bounds inference any compliant implementation may have.
We supply such a baseline in \S\ref{sec:bounds-inference_algorithm}, specified via a \emph{reference bounds engine}.
Bounds-inference implementations must be ``at least as good'' as the baseline in the following sense:

\begin{definition}[Bounds engine]\label{def:bounds-engine}
Let $P$ be a valid algorithm, $S$ a valid schedule, and $T = \mathcal{S}(S, \mathcal{L}(P))$.
A map $\beta : \kwd{Tgt}^? \to \kwd{Tgt}$ is a \emph{bounds engine} if $\beta(T) \in \kwd{BI}(T)$ for any such $T$.
\end{definition}

\begin{definition}[Bounds quality constraint]\label{def:bounds-quality}
Let $\beta_0$ be the reference bounds engine given in \S\ref{sec:bounds-inference_algorithm}.
Let $P$ be a valid algorithm, $S$ a valid schedule, and $z$ an input.
A bounds engine $\beta$ is compliant with the \emph{bounds quality constraint} if whenever (1) $P(z)$ does not contain an error value, (2) $\beta_0(T)(z)$ does not fail an assertion check, and (3) $\beta_0(T) \simeq_z P$, then $\beta(T) \simeq_z P$.
\end{definition}

Thus, any compliant implementation must produce a result on at least the set of programs and schedules accepted by the reference bounds algorithm.
In this way, programmers can be assured some degree of portability between different compliant implementations (or across versions of a single implementation).
For the existing Halide compiler, this reference method can be used to generate regression tests.

%% file: figures/overview/system.tex
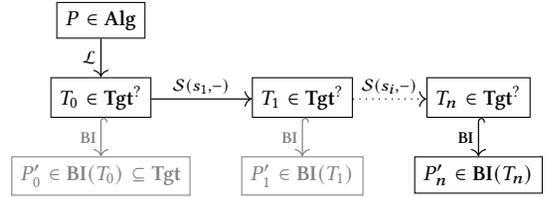
\begin{wrapfigure}{r}{0.51\textwidth}
\centering
\footnotesize
% https://ctan.math.washington.edu/tex-archive/graphics/pgf/contrib/tikz-cd/tikz-cd-doc.pdf
% https://tex.stackexchange.com/q/515267/42504
\begin{tikzcd}[cells={nodes={draw}}]
    P \in \kwd{Alg}       \arrow[d, "{\mathcal{L}}"'] & & \\
    T_0 \in \kwd{Tgt}^{?} \arrow[d, hookrightarrow, "{\kwd{BI}}"', gray]
                          \arrow[r, "{\mathcal{S}(s_1, -)}"] 
  & T_1 \in \kwd{Tgt}^{?} \arrow[d, hookrightarrow, "{\kwd{BI}}"', gray]
                          \arrow[r, dotted, "{\mathcal{S}(s_i, -)}"] 
  & T_n \in \kwd{Tgt}^{?} \arrow[d, hookrightarrow, "{\kwd{BI}}"'] \\
    |[gray]| P_0' \in \kwd{BI}(T_0) \subseteq \kwd{Tgt}
  & |[gray]| P_1' \in \kwd{BI}(T_1)
  & P_n' \in \kwd{BI}(T_n)
\end{tikzcd}
\caption{Compiling a program $P$ with schedule $S=s_1;\dots;s_n$ entails first lowering $P$ to an IR program $T_0$ and then applying each directive $s_i$ in turn. At the end, the bounds oracle $\kwd{BI}$ fills the holes in $T_n$ to produce a final program $P_n'$. Shown in gray are intermediate bounds oracle queries that are useful in the formal analysis, but are not part of compilation.}
\label{fig:overview-system}
\end{wrapfigure}

%% file: tex/alg_lang.tex
\input{figures/source/high-level-syntax}
Here we describe the Halide algorithm language, whose purpose is to define the values that the final, scheduled, program must compute.
It is a somewhat unusual dataflow language, consisting of \emph{funcs} whose values are computed on-demand by their dependents, and which might have one or more \emph{update stages}, which eagerly and in-place update the func being computed.
This scheme preserves referential transparency of funcs, but the resulting mix of eager and lazy semantics complicates any attempt to assign a simple denotation; this is why we use a big-step semantics.
Finally, the language is carefully designed with the scheduling language in mind: it underspecifies issues pertaining to bounds and evaluation order, while restricting the dependencies between funcs for the sake of analysis.

\subsection{Algorithm Terms and Expressions}\label{sec:alg_lang-syntax}\label{sec:alg_lang-expressions}
Our formalization of Halide (syntax in Figure~\ref{fig:source-syntax}) focuses on the fundamental issues at play: pure definitions, separable updates, and imperative updates.
Along with pointwise evaluation, these are the primary constructs that govern the structure of computation.

Programmers write \emph{pipelines} $P$, which are a sequence of \emph{func definitions} $F$. Each func has some dimension $n$, associated loop variables $x_1,\dots,x_n$, and a \emph{body} $B$. A func body is made up of one or more \emph{stages}. Each stage is made up of a \emph{reduction domain} (or ``rdom'') $R$, a predicate $e_P$, and a rule $(e_1,\dots,e_n) \leftarrow e$. The first stage $U_0=e$ is known as the \emph{pure stage} and is equivalent to $\kwd{rdom}() \kwd{in} ~(x_1,\dots,x_n) \leftarrow e \kwd{if} 1$.

A \emph{reduction domain} repeats the stage rule for a fixed list of variables $r_1,\dots,r_n$ (not necessarily of the same dimension as the func in which it appears) which range over provided intervals. These model a limited form of imperative updates on a func which happen before \emph{any} other func observes any of its values. The variable $r_1$ is innermost (changes fastest), while $r_n$ is outermost. As we will see in \S\ref{sec:alg_lang-rules}, there are many restrictions on the form of reduction domains and update stages.

Halide algorithms distinguish variables by their definition sites.
The variables that are bound by func definitions are lettered $x$ and are called \emph{pure variables}.
The variables bound by rdoms are lettered $r$ and are called \emph{reduction variables}.
Finally, variables bound by the top-level pipeline definitions are lettered $p$ are are called \emph{parameter variables}.

These parameters are optional and are always passed \emph{constants}, never other pipelines or funcs. A \emph{realization} $Z$ of a pipeline is a setting of the $m$ parameters, plus $n$ constant intervals over which to evaluate the output func, the last one in the pipeline, which is also named in its signature.

Figure~\ref{fig:expr-syntax} shows the syntax of the expression language.
The set of values $\mathbb{V} = \mathbb{Z} \cup \set{\varepsilon_{\mathrm{rdom}}, \varepsilon_{\mathrm{mem}}}$ in the formal language extends\footnote{The practical system also supports floating-point and fixed-width integers, and faces standard semantic issues with those.} $\mathbb{Z}$ with special \emph{error values}, which behave as follows:

\begin{definition}[Error value]\label{def:error-value}
The special expression values $\varepsilon_{\mathrm{rdom}} < \varepsilon_{\mathrm{mem}}$ encode a hierarchy of errors.
Any operation in the expression language involving one or more of these values evaluates to the \emph{greatest} among them.
\end{definition}

Note the omission of arithmetic errors in this definition.
These cannot arise because all operations in the expression language are \emph{total}.
In particular, division and modulo by zero are both defined to be zero.
The reason for this is discussed further in \S\ref{sec:conclusion-impact}.
The other errors, $\varepsilon_{\mathrm{rdom}}$ and $\varepsilon_{\mathrm{mem}}$, respectively capture errors preventing ordinary execution of rdoms (\S\ref{sec:alg_lang-semantics}) and \emph{memory} errors, which do not occur in the algorithm semantics.
Memory errors are possible in the target language (\S\ref{sec:target_lang-semantics}), but are prohibited by theorem \ref{thm:meta-safety}.

There is no Boolean type in the expression language, so the logical operators interpret their arguments according to the usual convention of using zero to represent ``false'' and non-zero values to represent ``true''. When a logical operator evaluates to ``true'', it returns $1$, specifically.

Finally, note that the expression language has no short-circuiting semantics.
Thus, logical-or and logical-and may not be used to conditionally evaluate points in another func, and the ``select'' function (the common ternary-if operator) always fully evaluates all three of its arguments.

\subsection{Algorithm Validity Rules}\label{sec:alg_lang-rules}
Halide algorithms must adhere to several non-standard restrictions.
This first rule constrains the use of pure variables to facilitate flexible scheduling decisions.

\begin{definition}[Syntactic separation restriction]\label{def:separation}
Let $f$ be a func given by $\kwd{fun} f(x_1,\dots,x_n) = \{ U_0 ; \cdots ; U_m \}$.
The \emph{syntactic separation restriction} states that for all pure variables $x_i$ and all stages $U_j$, if $x_i$ occurs anywhere in $U_j$ then all accesses in $U_j$ of the form $f[e_1,\dots,e_n]$ must have $e_i \equiv x_i$.
The update rule $U_j = R \kwd{in} ~(\dots,e_i,\dots) \leftarrow e \kwd{if} e_P$ must also have $e_i \equiv x_i$.
\end{definition}

This rule is critical to the correctness of many scheduling directives and metatheory claims, but it is quite subtle, so we show a few examples.
First, it might be tempting to write an in-place shift using the following func definition:
\[ \kwd{fun} f(x) = \{ g[x]; (x) \leftarrow f[x+1] \} \]
but such an update diverges on $f$'s unbounded domain since $f(0)$ would need to first compute $f(1)$, which would need to compute $f(2)$ and so on.
Such updates are disallowed by definition \ref{def:separation}.
It is also disallowed to use the variable in some places, but not others, as in:
\[ \kwd{fun} f(x) = \{ g[x]; \kwd{rdom}(r=(0,3)) \kwd{in} ~(x) \leftarrow f[x] + f[r] \} \]
The reason here is that, viewed as an in-place \emph{update} to the values of $f$, the update cannot be applied uniformly across the entire dimension $x$.
On the other hand, a definition like
\[ \kwd{fun} f(x) = \{ 0; \kwd{rdom}(r=(0,3)) \kwd{in} ~(x) \leftarrow f[x] + \textcolor{red}{g}[x] + \textcolor{red}{g}[r] \} \]
is legal since the restriction only applies to the func whose update stage is being defined.
At this point in the algorithm, all of $g$'s values are known, so there is no hazard.
Intuitively, updates that reference pure variables should augment the previous stage while remaining \emph{well-founded}.

The syntactic separation restriction extends the notion of purity from variables to stage \emph{dimensions}, which need not reference all of the func's pure variables.
\begin{definition}[Pure/reduction dimensions]\label{def:pure-dim}
For any pure variable $x_i$ and stage $U_j$, it is said that $i$ is a \emph{pure dimension} in stage $j$ if $x_i$ appears in $U_j$.
Dimensions which are not pure are called \emph{reduction dimensions}.
\end{definition}

Certain expressions in Halide may not refer to pure or reduction variables in order to keep scheduling flexible and sound.
Such expressions are called \emph{startup expressions} to reflect the fact that they are constant through the whole execution.

\begin{definition}[Startup expression]\label{def:startup-expression}
In a pipeline $P$ with parameters $p_1, \dots, p_n$, an expression $e$ is a \emph{startup expression} iff $e$ contains no func references and any variable $v$ occurring in $e$ is identically one of $p_i$ for some $i$.
\end{definition}

With this definition, we are finally ready to define validity for a program in the algorithm language.

\begin{definition}[Valid program]\label{def:valid}
A program $P\in\kwd{Alg}$ is \emph{valid} if the bounds of all rdoms are startup expressions, the names of all funcs are unique, the names of pure variables within each func are unique, and the names of reduction variables within a single stage are unique.
All stages must obey the syntactic separation restriction (definition \ref{def:separation}).
The output func $f$ in $\kwd{pipeline} f(\dots)$ must exist and be the last func defined.
All funcs must be defined before they are referenced by another func.
The first stage of every func may not include a self-reference (i.e., must be pure in all dimensions).
Lastly, common type checking rules for expressions (eg. func arity) must be respected.
\end{definition}

\subsection{Algorithm Semantics}\label{sec:alg_lang-semantics}
\input{figures/source/semantics}
The purpose of a Halide algorithm is to \emph{define} the value of every point in every func (Figure~\ref{fig:source-semantics}).
Evaluation proceeds pointwise with no need to track bounds.
Funcs are evaluated by substitution [Func-Eval] as is standard for function calls.
Compared to the target language, which precomputes values of funcs as if they were arrays, these semantics are \emph{lazy}.

While this laziness avoids reasoning about bounds, it complicates the semantics of the comparatively eager rdom construct.
How do we update a func in-place, when it is intuitively meant to be pure?
To resolve this tension we, simply unroll rdoms [RDom-Eval] into sequences of point updates when and as they are encountered.

These simple point updates [Update-Eval] can then be thought of as shadowing the previous func definition, similar to the functional definition of stores used by most operational semantics for imperative languages.
If the lookup point and update point coincide, then the update rule is substituted, otherwise the existing value is used.

Lastly, we note that all valid algorithms terminate.
This follows the intuition that Halide pipelines are defining mathematical objects by supplying formulas to compute the values.

\begin{lemma}[Algorithms terminate]
Given any algorithm $P \in\kwd{Alg}$ and input $z$, the output of $P(z)$ can be determined in a finite amount of time.
\end{lemma}

\begin{proof}
Since rdom bounds are \emph{startup expressions} (\ref{def:startup-expression}) and no infinity value exists in $\mathbb{V}$, there is no way to loop infinitely.
The program validity checks (\ref{def:valid}) prevent self-recursion in the function definitions. 
Functions must be declared before they are used, so recursion is impossible.
Thus, Halide algorithms always terminate.
In fact, this also shows Halide is not Turing-complete.
\end{proof}

%% file: figures/source/high-level-syntax.tex
\begin{figure}
\begin{subfigure}[b]{0.5\linewidth}
\footnotesize
\begin{tabularx}{0.98\linewidth}{rcXl}
  \toprule
  $R$            & $::=$ & $\kwd{rdom}(r_1 = I_1 , \dots , r_k = I_k)$       & rdom \\
  $U_0$          & $::=$ & $e$                                               & pure stage \\
  $U$            & $::=$ & $R \kwd{in} ~(e_1, \ldots, e_n) \leftarrow e      %
                           \kwd{if} e_P$                                     & update stage \\
  $B$            & $::=$ & $U_0; U*$                                         & func body \\
  $F$            & $::=$ & $\kwd{fun} f(x_1, \dots, x_n) = \{ B \}$          & func \\
  $D$            & $::=$ & $F ~|~ D; F$                                      & definitions \\
  $P$            & $::=$ & $\kwd{pipeline} f(p_1,\dots,p_m) = D$             & pipeline \\
  \midrule
  $Z$            & $::=$ & $P; \kwd{realize}(z)$                             & realization \\
  $z$            & $::=$ & $\tup{(I_1, \dots, I_n), (c_1, \dots, c_m)}$      & input \\
  \bottomrule
\end{tabularx}
\Description{Algorithm language}
\caption{Algorithm language}
\label{fig:source-syntax}
\end{subfigure}%
\begin{subfigure}[b]{0.5\linewidth}
\footnotesize
\begin{tabularx}{0.98\linewidth}{rcXl}
  \toprule
  $c$        & $::=$ & $i \in \mathbb{V}$                                & constants   \\
  $a$        & $::=$ & $f[e_1,\ldots,e_n]$                               & func access \\
  $e$        & $::=$ & $c ~|~ a ~|~ v ~|~ \kwd{op}(e_1, \dots, e_k)$     & expression  \\
  $I$        & $::=$ & $(e_{\min}, e_{\len})$                            & interval    \\
  \midrule
  $v$        & $::=$ & $x$                                               & pure variable      \\
             & ~|~   & $r$                                               & reduction variable \\
             & ~|~   & $p$                                               & parameter variable \\
  $\kwd{op}$ & $::=$ & $+ ~|~ - ~|~ \times ~|~ \Div ~|~ \Mod$            & arithmetic         \\
             & ~|~   & $\lor ~|~ \land ~|~ \neg ~|~ < ~|~ > ~|~ =$       & logical            \\
             & ~|~   & $\select ~|~ \min ~|~ \max$                       & conditional        \\
  \bottomrule
\end{tabularx}
\Description{Expression language}
\caption{Expression language}
\label{fig:expr-syntax}
\end{subfigure}
\Description{Halide high-level syntax}
\caption{High-level Halide syntax descriptions. These capture the core components of the
user-facing algorithm and expression languages in the production Halide language.}
\end{figure}

%% file: figures/source/semantics.tex
\begin{figure*}
\footnotesize
\begin{framed}
    \AxiomC{$\forall y=(i_1,\dots,i_n)\in I_1 \times \dots \times I_n: \left[\svec{c/p}\right] \left(D; f[i_1,\dots,i_n]\right) \bigeval c_{y}$}
    \RightLabel{[Realize]}
    \UnaryInfC{$\kwd{pipeline} f(\svec{p})=D; \kwd{realize}\left(\svec{I}, \svec{c}\right) \bigeval g(y) := c_{y}$}
    \DisplayProof
  \\\vspace{1ex}
    \AxiomC{}
    \RightLabel{[Const-Eval]}
    \UnaryInfC{$D ; c \bigeval c$}
    \DisplayProof
  \quad
    \AxiomC{$\left(\svec{D ; e \bigeval c}\right)$}
    \AxiomC{$D ; f[\svec{c}] \bigeval c'$}
    \RightLabel{[Func-Arg-Eval]}
    \BinaryInfC{$D ; f[\svec{e}] \bigeval c'$}
    \DisplayProof
  \\\vspace{1ex}
    \AxiomC{$\left(\svec{D ; e \bigeval c}\right)$}
    \AxiomC{$D ; \kwd{op}(\svec{c}) \bigeval c'$}
    \RightLabel{[Op-Eval]}
    \BinaryInfC{$D ; \kwd{op}(\svec{e}) \bigeval c'$}
    \DisplayProof
  \quad
    \AxiomC{$f \neq g$}
    \AxiomC{$D ; f[\svec{c}] \bigeval c'$}
    \RightLabel{[Func-Skip]}
    \BinaryInfC{$D ; \kwd{fun} g[\svec{x}] = \{ B \}; f[\svec{c}] \bigeval c'$}
    \DisplayProof
  \quad
    \AxiomC{$D ; \left[\svec{c/x}\right] e \bigeval c'$}
    \RightLabel{[Func-Eval]}
    \UnaryInfC{$D ; \kwd{fun} f[\svec{x}] = \{ e \} ; f[\svec{c}] \bigeval c'$}
    \DisplayProof
  \\\vspace{1ex}
    \AxiomC{$D ; \kwd{fun} f[\svec{x}] = \{ \svec{U} \} ;
                 \left[\svec{c/x}\right] \left( \kwd{select}(
                 \svec{e = x} \land e_p, e_b, f[\svec{x}]) \right) \bigeval c'$}
    \RightLabel{[Update-Eval]}
    \UnaryInfC{$D ; \kwd{fun} f[\svec{x}] = \{ \svec{U} ;
                 \kwd{rdom}() \kwd{in} ~(\svec{e}) \leftarrow e_b \kwd{if} e_p \} ;
                 f[\svec{c}] \bigeval c'$}
    \DisplayProof
  \\\vspace{1ex}
    \AxiomC{$\left(\svec{I=\tup{e^{\kwd{min}}, e^{\kwd{len}}}}\right)$}
    \AxiomC{$\left(\svec{e^{\kwd{min}} \bigeval c^{\kwd{min}}}\right)$}
    \AxiomC{$\left(\svec{e^{\kwd{len}} \bigeval c^{\kwd{len}}}\right)$}
    \AxiomC{$\exists j. c_j^{\kwd{len}} < 0$}
    \RightLabel{[RDom-Err]}
    \QuaternaryInfC{$D ; \kwd{fun} f[\svec{x}] = \{ \cdots ; \kwd{rdom}(\svec{r=I}) \kwd{in} \cdots \} ; f[\svec{c}] \bigeval \varepsilon_{\mathrm{rdom}}$}
    \DisplayProof  
  \\\vspace{1ex}
    \AxiomC{\hspace{-0.0em}$\left(\svec{I=\tup{e^{\kwd{min}}, e^{\kwd{len}}}}\right)$}
    \AxiomC{\hspace{-0.3em}$\left(\svec{e^{\kwd{min}} \bigeval c^{\kwd{min}}}\right)$}
    \AxiomC{\hspace{-0.3em}$\left(\svec{e^{\kwd{len}} \bigeval c^{\kwd{len}}}\right)$}
    \AxiomC{\hspace{-0.3em}$D ; \kwd{fun} f[\svec{x}] = \{ \svec{U} ; \mathit{unroll} \} ; f[\svec{c}] \bigeval c'$}
    \RightLabel{[RDom-Eval]}
    \QuaternaryInfC{$D ; \kwd{fun} f[\svec{x}] = \{ \svec{U} ; \kwd{rdom}(\svec{r=I}) \kwd{in} ~(\svec{e}) \leftarrow e_b \kwd{if} e_p \} ; f[\svec{c}] \bigeval c'$}
    \DisplayProof

  \vspace{1ex}
  \emph{where}  
  \vspace{1ex} 
  
    $\mathit{unroll} =
    \left(
        \left[ c_k^{\kwd{min}} / r_k \right]
          \left( \kwd{rdom}(r_1=I_1, \dots, r_{k-1}=I_{k-1}) 
                 \kwd{in} ~(\svec{e}) \leftarrow e_b \kwd{if} e_p \right) ;
        \quad \cdots \quad
        \left[ \left(c_k^{\kwd{min}} + c_k^{\kwd{len}} - 1\right) / r_k \right]
           \left( \cdots \right)
    \right)$
\end{framed}
\Description{Algorithm language natural semantics}
\caption{Algorithm language natural semantics. Note that we use a metasyntactic notation,
$\svec{~\cdot~}$, which indicates that the covered expression is repeated for each numerical
subscript, eg. $(\svec{e=x} \land e_p) \equiv (e_1=x_1 \land \dots \land e_n=x_n \land e_p)$.
Parentheses distinguish such terms from axioms. The [Realize] rule defines the points of a
partial function $g : (I_1 \times \dots \times I_n) \to \mathbb{V}$.}
\label{fig:source-semantics}
\end{figure*}

%% file: tex/target_lang.tex
In this section, we describe the target language (IR) to which the algorithm language compiles.
Unlike the algorithm language, it is similar to classic imperative languages, and programs in this language have a defined execution order (which is modified by the schedule).
It uses the same expression language from \S\ref{sec:alg_lang-expressions} and has the same semantics for all expressions, save func accesses, which become references to memory.

\subsection{Syntax}\label{sec:target_lang-syntax}
\input{figures/target/syntax}
Figure~\ref{fig:target-syntax} presents the abstract syntax for the Halide IR.
This language comes in two variants: \emph{with} holes ($\kwd{Tgt}^?$) and \emph{without} holes ($\kwd{Tgt}$).
The lowering algorithm given in \S\ref{sec:lowering} translates an algorithm to a program in $\kwd{Tgt}^?$ whose holes will be filled by \emph{bounds inference} (\S\ref{sec:bounds}).
The main difference between $\kwd{Tgt}$ and similar imperative languages is that loops are restricted to range-based for loops which can be marked for parallel traversal.
Furthermore, these ranges are given as \emph{minimum} and \emph{length} pairs, rather than minimum and \emph{maximum}.
Some syntax may be annotated with \emph{labels}, written $\ell$.
Labels are ignored by the semantics because they are simply used as handles by the \emph{scheduling} (\S\ref{sec:scheduling}) and bounds inference systems.

\subsection{Semantics}\label{sec:target_lang-semantics}
\input{figures/target/semantics}
In Figure~\ref{fig:target-semantics} we give small-step semantics for the IR.
Note that $\Sigma$ is an \emph{environment} for loop variables and let bindings and $\sigma$ is the \emph{store} or \emph{heap} in which memory is allocated.

These semantics are mostly standard, though there are a few instances where the semantics can get \emph{stuck}. We enumerate and define all the failure modes here:
\begin{definition}[Assertion failure]\label{def:assertion_failure}
If the execution of a program $P\in\kwd{Tgt}$ gets stuck when an assertion fails (i.e. the condition evaluates to $0$), then we say $P$ has failed an assertion check.
\end{definition}
\begin{definition}[RDom failure]\label{def:rdom_failure}
If the execution of a program $P\in\kwd{Tgt}$ gets stuck because a for loop has a negative extent, then we say $P$ has encountered an \emph{rdom failure}. This corresponds to the failure mode in the algorithm semantics (\S\ref{sec:alg_lang-semantics}) where an invalid rdom causes the program to return $\varepsilon_\text{rdom}$ everywhere.
\end{definition}
\begin{definition}[Memory error]\label{def:memory_failure}
Recall that the [Read] and [Assn] rules assume their accesses are in bounds. If the execution of a program $P\in\kwd{Tgt}$ gets stuck when accessing memory, we say $P$ has attempted an \emph{out of bounds} access or has encountered a \emph{memory error}.
\end{definition}
Recall that theorem \ref{thm:meta-safety} states that memory errors cannot occur in the execution of a program which was derived from an algorithm via lowering, scheduling, and bounds inference.

The [Alloc] rule updates the store $\sigma$ with a mapping from the \emph{symbolic name} of the func to a pair of (1) a partial function $\hat{f}$ (initially $\varepsilon_\mem$ everywhere) that records the values and (2) the bounds that were stated at allocation time.
The predicate $\mathrm{InBounds}(f[\svec{c}], \sigma)$ uses this data to check the fully evaluated point $\svec{c} \equiv (c_1, \dots, c_n)$ against the bounds stored in $\sigma(f)$.

[Assn] defines assigning to a point in a func in the store and [Read] defines reading from a func in the store.
Assignment is modeled by \emph{shadowing} the old value, ie. by redefining the mapping of $f$ in $\sigma$ to a new partial function $\hat{f}'$ which agrees with $\hat{f}$ everywhere except at the point being updated.
We use the terse syntax $\sigma'=\sigma[f(\svec{c})=c']$ to denote this operation.
Reading a value from a func is then a matter of simply evaluating the stored function.

%% file: figures/target/syntax.tex
\begin{wrapfigure}{R}{0.51\textwidth}
\footnotesize
\begin{tabularx}{\linewidth}{rcXl}
  \toprule
  $\tau$         & $::=$ & $\mathrm{serial} ~|~ \mathrm{parallel}$           & traversal order \\
  $s$            & $::=$ & $\kwd{nop}$                                       & no operation \\
                 & ~|~   & $\kwd{assert} e$                                  & assertion \\
                 & ~|~   & $s_1$ ; $s_2$                                     & sequencing \\
                 & ~|~   & $\kwd{allocate} f(I_1, \ldots, I_n)$              & allocate buffer \\
                 & ~|~   & $f^\ell[e_1,\dots,e_m] \leftarrow e_0$            & update buffer \\
                 & ~|~   & $\kwd{if} e_1 \kwd{then} s_1 \kwd{else} s_2$      & branching \\
                 & ~|~   & $\kwd{for}^\tau x \kwd{in} I \kwd{do} s$          & bounded loops \\
                 & ~|~   & $\kwd{let} x = e \kwd{in} s$                      & let binding \\
                 & ~|~   & $\kwd{label} \ell : s$                            & statement label \\
  $P$            & $::=$ & $\kwd{pipeline} f(p_1,\dots,p_m): ~s$             & pipeline \\
  \midrule
  % TODO: ?\ell needs to be fleshed out into ?^{cpu/mem} f^{i,j}_k
  $e$            & $::=$ & $\dots ~|~ ?\ell ~|~ f^\ell[\dots]$             & $\kwd{Tgt}^?$ expr \\
  \bottomrule
\end{tabularx}
\Description{Halide IR syntax}
\caption{Halide IR syntax. Expressions and realizations are the same as in Figure~\ref{fig:expr-syntax}, but are augmented with labeled holes for $\kwd{Tgt}^?$. Labels $\ell$ are arbitrary and left uninterpreted. Statement labels have no special semantics.}
\label{fig:target-syntax}
\end{wrapfigure}

%% file: figures/target/semantics.tex
\begin{figure*}
\footnotesize
\newcommand{\rulesep}{\\\vspace{0.75em}}
\newcommand{\ruleskip}{\hspace{1.5em}}
\begin{mdframed}
\fbox{\begin{tabular}{r@{\hskip 1em}c@{\hskip 1em}l}
$E$ & ::= & $\square$                                                          \\
    & |   & $E ; s$                                                            \\
    & |   & $\kwd{allocate} f((c^{\min}_1, c^{\len}_1), \dots, E, \dots, I_n)$ \\
    & |   & $a \leftarrow E$                                                   \\
    & |   & $E \leftarrow c$                                                   \\
    & |   & $\kwd{if} E \kwd{then} s_1 \kwd{else} s_2$                         \\
    & |   & $\kwd{for}^\ell x \kwd{in} E \kwd{do} s$                           \\
    & |   & $\kwd{let} x = E \kwd{in} s$                                       \\
    & |   & $f[c_1, \ldots, E, \ldots, e_n]$                                   \\
    & |   & $\kwd{op}(c_1, \ldots, E, \ldots, e_k)$                            \\
    & |   & $(E, e^{\len})$                                                    \\
    & |   & $(c^{\min}, E)$                                                    \\
    & |   & $\kwd{assert} E$                                                   \\
    & |   & $\kwd{pipeline} f(\svec{p}): E$                                    \\
\end{tabular}}%
\hfill%
\begin{minipage}{0.55\textwidth}
\centering
\AxiomC{$\tup{s | \Sigma | \sigma} \to \tup{s' | \Sigma' | \sigma'}$}
\RightLabel{[Reduce]}
\UnaryInfC{$\tup{E[s] | \Sigma | \sigma} \to \tup{E[s'] | \Sigma' | \sigma'}$}
\DisplayProof
\rulesep
\AxiomC{$v = \Sigma(x)$}
\RightLabel{[Var]}
\UnaryInfC{$\tup{x | \Sigma | \sigma} \to \tup{v | \Sigma | \sigma}$}
\DisplayProof
\rulesep
\AxiomC{\phantom{M}}
\RightLabel{[Nop]}
\UnaryInfC{$\tup{\kwd{nop} ; s | \Sigma | \sigma} \to \tup{s | \Sigma | \sigma}$}
\DisplayProof
\rulesep
\AxiomC{$c \neq 0$}
\RightLabel{[Assert-True]}
\UnaryInfC{$\tup{\kwd{assert} c | \Sigma | \sigma} \to \tup{\kwd{nop} | \Sigma | \sigma}$}
\DisplayProof
\rulesep
\AxiomC{\phantom{M}}
\RightLabel{[Let]}
\UnaryInfC{$\tup{\kwd{let} x = c \kwd{in} s | \Sigma | \sigma} \to \tup{s | \Sigma[x=c] | \sigma}$}
\DisplayProof
\end{minipage}
%%%%%%%%%%%%%%%%%%%%%%%%%%%%%%%%%%%%%%%%%%%%%%%%%%%%%
\begin{center}
\AxiomC{\phantom{M}}
\RightLabel{[Realize]}
\UnaryInfC{$\tup{\kwd{pipeline} f(\svec{p}): s ; \kwd{realize}(\svec{c}) | \Sigma | \sigma} \to \tup{\kwd{pipeline} f(\svec{p}): s | \Sigma\left[\svec{p=c}\right] | \sigma}$}
\DisplayProof
\rulesep
\AxiomC{$\sigma(f) = \tup{\hat{f}, \ldots}$}
\RightLabel{[End]}
\UnaryInfC{$\tup{\kwd{pipeline} f(\svec{p}): \kwd{nop} | \Sigma | \sigma} \to \tup{\hat{f} | \Sigma | \sigma}$}
\DisplayProof
\rulesep
\AxiomC{$c^{\len} > 0$}
\RightLabel{[For-Iter]}
\UnaryInfC{$\tup{\kwd{for} x \kwd{in} ~(c^{\min}, c^{\len}) \kwd{do} s | \Sigma | \sigma} \to \tup{s ; \kwd{for} x \kwd{in} ~(c^{\min}+1, c^{\len}-1) \kwd{do} s | \Sigma[x=c^{\min}] | \sigma}$}
\DisplayProof
\rulesep
\AxiomC{\phantom{M}}
\RightLabel{[For-Stop]}
\UnaryInfC{$\tup{\kwd{for} x \kwd{in} ~(c^{\min}, 0) \kwd{do} s | \Sigma | \sigma} \to \tup{\kwd{nop} | \Sigma \setminus\!\set{x} | \sigma}$}
\DisplayProof
\rulesep
\AxiomC{$c\neq 0$}
\RightLabel{[If-T]}
\UnaryInfC{$\tup{\kwd{if} c \kwd{then} s_1 \kwd{else} s_2 | \Sigma | \sigma} \to \tup{s_1 | \Sigma | \sigma}$}
\DisplayProof
\ruleskip
\AxiomC{\phantom{M}}
\RightLabel{[If-F]}
\UnaryInfC{$\tup{\kwd{if} 0 \kwd{then} s_1 \kwd{else} s_2 | \Sigma | \sigma} \to \tup{s_2 | \Sigma | \sigma}$}
\DisplayProof
\rulesep
\AxiomC{$\mathrm{InBounds}(f[\svec{c}], \sigma)$}
\AxiomC{$\sigma' = \sigma[f(\svec{c}) = c']$}
\RightLabel{[Assn]}
\BinaryInfC{$\tup{f[\svec{c}] \leftarrow c' | \Sigma | \sigma} \to \tup{\kwd{nop} | \Sigma | \sigma'}$}
\DisplayProof
\ruleskip
\AxiomC{$\sigma' = \sigma\left[f \to \tup{\lambda \svec{x}. ~\varepsilon_\mem, ~\svec{I}}\right]$}
\RightLabel{[Alloc]}
\UnaryInfC{$\tup{\kwd{allocate} f\left(\svec{I}\right) | \Sigma | \sigma} \to \tup{\kwd{nop} | \Sigma | \sigma'}$}
\DisplayProof
\rulesep
\AxiomC{$\mathrm{InBounds}(f[\svec{c}], \sigma)$}
\RightLabel{[Read]}
\UnaryInfC{$\tup{f[\svec{c}] | \Sigma | \sigma} \to \tup{\sigma(f[\svec{c}]) | \Sigma | \sigma}$}
\DisplayProof
\ruleskip
\AxiomC{$\kwd{op}(\svec{c}) = c'$}
\RightLabel{[Eval]}
\UnaryInfC{$\tup{\kwd{op}(\svec{c}) | \Sigma | \sigma} \to \tup{c' | \Sigma | \sigma}$}
\DisplayProof
\end{center}
\end{mdframed}
\Description{IR structural semantics}
\caption{Structural semantics for $\kwd{Tgt}$ (\emph{without} holes). Notice that there are four states that can get \emph{stuck}: (1) when $\kwd{assert} \kwd{false}$ is encountered, (2) when a $\kwd{for}$ loop extent is negative, (3) when a read occurs out of bounds, and (4) when an assignment occurs out of bounds. The latter two memory errors cannot happen in programs derived from the scheduling and bounds inference processes by Theorem~\ref{thm:meta-safety}.}
\label{fig:target-semantics}
\end{figure*}

%% file: tex/lowering.tex
\input{figures/lowering}

Halide algorithms are compiled to IR programs with holes by the \emph{lowering} function $\mathcal{L}$, defined in Figure~\ref{fig:lowering}.
The lowering function creates a sequence of top-level loop nests for every func in the program.
Inside these loops are assignments implementing the formulas for each stage in the algorithm.
Pure dimensions which do not appear in a stage are not lowered, and reduction domains appear as innermost loops.

The lowering function also annotates certain fragments with \emph{labels} to facilitate scheduling and bounds inference.
These labels appear in three places: first, they appear in $\kwd{label}$ statements which act as handles for the scheduling directives; second, they are attached to the $\cpu$ and $\mem$ \emph{bounds holes}; finally, they are attached to func references.
The following lemma captures the structural invariant provided by the first set of these labels.

\begin{lemma}[Loop naming]\label{lem:loop-naming}
Given a valid algorithm $P\in\kwd{Alg}$ and a valid schedule $S\in\kwd{Sched}$, any for loop in $\mathcal{S}(S,\mathcal{L}(P))$ is uniquely identified by (1) the func, (2) the specialization (or lack thereof, see \S\ref{sec:scheduling-specialization}), and (3) the stage to which it belongs, as well as (4) the name of its induction variable.
\end{lemma}

Specializations do not exist in initially lowered programs, but are a scheduling feature (see \S\ref{sec:scheduling-specialization}) that enables replicating code behind one or more \emph{branches}, each guarded by a predicate.
Each branch can be scheduled independently, and its predicate is used to simplify the body.
A common use case is to specialize a pipeline to common input sizes and reduce bounds computations.
If a func is not specialized, that data can be regarded as $0$.
In any case, Lemma~\ref{lem:loop-naming} lets us relate syntax fragments in the IR to their \emph{provenance} in the original algorithm.
The following lemma uses this to state that funcs are computed and allocated in a valid order in the IR.

\begin{lemma}[Dominance]\label{lem:dominance}
Let $P\in\kwd{Alg}$ and $S\in\kwd{Sched}$ be a valid algorithm and schedule, and let $P' = \mathcal{S}(S,\mathcal{L}(P))$.
If a func $f$ appears in the definition of a func $g$ in $P$, then the loops for $f$ \emph{dominate} the assignment statement for $g$ in $P'$.
Furthermore, the $\kwd{allocate}$ statement for any func $f$ dominates the loops for $f$ in $P'$.
\end{lemma}

The previous two lemmas hold just after lowering by construction.
Each scheduling directive needs to show that it maintains these invariants.
Lowering also introduces a set of labeled \emph{bounds holes}, which will be filled by the \emph{bounds inference} oracle (\S\ref{sec:bounds}), and which carry the following data.

\begin{definition}[Bounds hole]\label{def:bounds-hole}
A bounds hole is an entity in the expression language of $\kwd{Tgt}^?$ that stands in for a hole-free expression.
A bounds hole is labeled by (1) whether it is an allocation hole ($\mem$) or a compute hole ($\cpu$), (2) whether it represents the minimum ($\min$) of an interval, or its length ($\len$), (3) the associated func and dimension, and (4) if it is a compute hole, the associated stage and specialization.

Across specializations, the last stage of a given func always uses a common bounds hole.
We omit the stage number when referring to the \emph{last} stage of a func and we omit the specialization number when the func is not specialized.
Finally, we write $?^{\mem}f_x = \left[(?^{\mem}f_x)^{\min}, (?^{\mem}f_x)^{\len}\right]$ for the allocation bounds interval for func $f$, dimension $x$.
By analogy, $?^{\cpu}f_x^{i,j}$ denotes the compute bounds interval for func $f$, dimension $x$, stage $i$, and specialization $j$.
\end{definition}

Finally, the labels attached to func references ($f^\ell[\dots]$) record the \emph{previous} stage and current specialization. This helps the bounds extraction procedure (\S\ref{sec:bounds}) construct the necessary predicates to ensure safety and correctness.

%% file: figures/lowering.tex
\begin{figure}
\footnotesize
\newcommand{\LL}{\mathcal{L}} % figures introduce groups. \LL not visible outside figure
\begin{framed}
  \begin{align*}
  \LL(\kwd{pipeline} f(\svec{p}) : \svec{F}; \kwd{fun} f(\svec{x}) = B) &= 
    \begin{cases}
      \kwd{pipeline} f\left(\svec{p}, \svec{x^{\min}, x^{\len}}\right) : & \\
      \quad \svec{\LL(F)}; \LL(\kwd{fun} f(\svec{x}) = B) &
    \end{cases} \\
  \LL(\kwd{fun} f(\svec{x}) = B) &= \begin{cases}
  \kwd{allocate} f\left(\svec{?^{\mem}f_x}\right) &\\
  \kwd{label} f: \LL_B(f, \svec{x}, B) &\\
  \end{cases} \\
  \LL_B(f, \svec{x}, U_0; \cdots; U_m) &= \begin{cases}
  \kwd{label} s_0: \LL_U(f, \svec{x}, 0, U_0) &\\
  \quad \vdots &\\
  \kwd{label} s_m: \LL_U(f, \svec{x}, m, U_m) &\\
  \end{cases} \\
  \LL_U(f, \svec{x}, i, R \kwd{in} \svec{e} = e_B \kwd{if} e_P) 
    &= \LL_P(f, \svec{x}, \svec{e}, i, \LL_R(R, \kwd{if} e_P \kwd{then} f^i[\svec{e}] \leftarrow e_B)) \\
  \LL_R(\kwd{rdom}(), s) &= s \\
  \LL_R(\kwd{rdom}(r_1=I_1, \svec{r=I}), s) &= \LL_R(\kwd{rdom}(\svec{r=I}), \kwd{for} r_1 \kwd{in} I_1 \kwd{do} s) \\
  \LL_P(f, (), (), i, s) &= [f^{i-1}/f] s \\
  \LL_P(f, (x, \svec{y}), (e, \svec{e'}), i, s) &=
        \begin{cases}
          \LL_P\left(
              f, \svec{y}, \svec{e'}, i, \kwd{for} x \kwd{in} ~?^{\cpu}f_x^i \kwd{do} s
          \right) & \text{if} ~ x \equiv e \\
          \LL_P(f, \svec{y}, \svec{e'}, i, s) & \text{otherwise} \\
        \end{cases}
  \end{align*}
\end{framed}
\Description{Lowering algorithm with default eager schedule.}
\caption{Lowering algorithm with default eager schedule.}
\label{fig:lowering}
\end{figure}

%% file: tex/bounds.tex
\input{figures/bounds/combined}

Previous work on Halide discusses bounds inference in terms of a particular algorithm used to fill the bounds holes.
Improvements to the compiler regularly change the results of this algorithm, resulting in an unstable definition in practice.

In order to abstract over the ever-changing bounds inference algorithm, we pose bounds inference as a \emph{program synthesis} problem via an oracle query.
While the resulting satisfiability problem is undecidable in general, this definition provides previously underformulated soundness conditions for any bounds inference algorithm. Queries to this oracle are defined as follows:

\begin{definition}[Bounds oracle query]\label{def:oracle-query}
Let $P\in\kwd{Alg}$ be an algorithm and let $S\in\kwd{Sched}$ be a schedule for it so $T = \mathcal{S}(S, \mathcal{L}(P))$.
Then a query to the \emph{bounds oracle} $\mathcal{O}$ is the \emph{predicate} $p = \mathcal{B}(T)$.
The oracle responds with some set of hole substitutions $\Gamma \in \mathcal{O}(p)$ that is compatible with $T$.
Hence, the set $\kwd{BI}(T) = \set{[\Gamma] T | \Gamma \in \mathcal{O}(\mathcal{B}(T))}$.
\end{definition}

Recall from definition~\ref{def:bounds-hole} that there are two kinds of bounds.
The \emph{compute bounds} define regions over which the points in the buffers must have non-error values that agree with those defined by the original algorithm.
The \emph{allocation bounds} enclose the compute bounds, and further includes at least all points read from or written to.
As we saw in the example (\S\ref{sec:example}), this gap can be exploited by overcompute strategies during scheduling (\S\ref{sec:scheduling-loops}).

\subsection{Bounds constraint extraction}\label{sec:bounds-extraction}
The algorithm for extracting the bounds constraints for a program $T\in\kwd{Tgt}^?$ is shown in Figure~\ref{fig:bounds-extraction}.
The extraction traverses the AST of the program and translates every statement into a logical condition with existentially quantified holes.

This extraction encodes a few important correctness conditions.
First, if a point being computed lies in the compute bounds, then all of the accesses on the right hand side of the assignment must be in the compute bounds of their funcs.
(What happens outside the compute bounds stays outside the compute bounds.)
Second, accesses occurring anywhere inside an expression that is used for \emph{indexing} or \emph{branching} must be in the compute bounds as well.
Finally, every point that is read anywhere in the program must at least be in the allocation bounds, in order to preserve memory safety.

This second point is particularly important: splitting loops in data-dependent update stages (such as when computing a histogram) will introduce $\kwd{if}$ statements whose values must not be errors resulting from reading uninitialized memory. The rule for $\kwd{let}$ is similarly motivated; $\kwd{let}$ expressions are only introduced by scheduling directives to hold expressions used for indexing (\S\ref{sec:scheduling}), so accesses there must be in the compute bounds.

\subsection{Reference algorithm}\label{sec:bounds-inference_algorithm}
Figure~\ref{fig:bounds-inference} gives the baseline bounds inference algorithm $\beta_0$.
It works by scanning the extracted constraint and performing interval arithmetic (via $\mathcal{I}$) on the terms, naively trying to symbolically satisfy the consequent of each implication without using its predicate (ie. unconditionally).
$\beta_0$ merges these intervals to determine safe coverings for each hole.
Because the constraint is extracted from the fully scheduled target program, it can rely on the association order of $\land$ to reflect the sequencing order in the original program and ensure that we make inferences about holes backwards through the dependencies.
Since $\beta_0$ only produces a list of substitutions, it does not meet the bounds engine definition (\ref{def:bounds-engine}) on its own.
However, it is easily lifted to a bounds engine by applying the substitutions whenever every hole is determined and no $\pm\infty$ appears in the substitutions.
When this is not the case, it simply fails by replacing the body with $\kwd{assert} \kwd{false}$.

Beyond the na\"{i}vety of the algorithm, interval arithmetic has an inherent \emph{dependency problem}.
The classic example is $x^2+x$, where $x\in[-1,1]$ and so $x^2\in[0,1]$. 
Adding these bounds gives $[-1,2]$, which is slightly wider than the true bounds: $[-\frac{1}{4}, 2]$.
This is because interval arithmetic treats $x^2+x$ as $x^2+y$, where $y$ varies \emph{independently} over the same interval as $x$.
These errors can accumulate rapidly as expressions grow larger.

$\beta_0$ is only meant to be a \emph{baseline}; and, although this algorithm is quite na\"{i}ve, it still identifies tight bounds for the example in \S\ref{sec:example}.
The practical system contains many improvements over this, including analyses of function value ranges, of correlated differences and sums, and of simplifications based on scoped facts.

\subsection{Metatheory}
Finally, we state the main lemmas concerning the structure of solutions to the bounds inference problem.

\begin{lemma}[Memory safety]\label{lem:bounds-memory-safety}
All programs resulting from bounds inference $P'\in\kwd{BI}(T)$, are memory safe.
\end{lemma}

\begin{lemma}[Compute bounds confluent]\label{lem:compute-confluency}
Let $P\in\kwd{Alg}$, $P'\in\kwd{BI}(\mathcal{L}(P))$, and let $f$ be a func in $P$.
If all of the points in compute bounds of funcs preceding $f$ are confluent with $P$, then the loop nest for $f$ computes values confluent with $P$.
\end{lemma}

The proofs of these lemmas are deferred to the appendices.
Together, they form the base case of the inductive proof that the scheduling directives are sound.

%% file: figures/bounds/combined.tex
\begin{figure}
\footnotesize
\begin{framed}
\begin{subfigure}[b]{0.55\textwidth}
\centering
\newcommand{\BB}{\mathcal{B}}
  \begin{align*}
    \BB(\kwd{pipeline} f(\svec{p}, &\svec{x^{\min}, x^{\len}}) : s) = \ldots \\
    \svec{\forall p\in[-\infty,\infty]}: &\BB(s) \land \svec{[x^{\min}, x^{\len}] \in ?^{\cpu}f} \\
    \BB(\kwd{assert} e) &= e \\
    \BB(\kwd{nop}) &= \kwd{true} \\
    \BB(\kwd{allocate} f(\dots) ; s) &= \exists ?^{\mem}f : \BB(s) \\
    \BB(s_1 ; s_2) &= \BB(s_1) \land \BB(s_2) \\
    \BB(\kwd{label} f : s) &= \exists ?^{\cpu}f: \BB(s) \\
    %%
    % Subtle: lets are always indices (only introduced by split), so they must be non-garbage (valid)
    \BB(\kwd{let} v = e \kwd{in} s) &= \BB_\cpu(e) \land \kwd{let} v = e \kwd{in} \BB(s) \\
    \BB(\kwd{if} e \kwd{then} s_1 \kwd{else} s_2)
        &= \BB_\cpu(e) \land e \implies \BB(s_1) \land \lnot e \implies \BB(s_2)
  \end{align*}
  \[ \BB(\kwd{for} v \kwd{in} [e^{\min}, e^{\len}] \kwd{do} s) = e^{\len} \geq 0 \land \forall v \in [e^{\min}, e^{\len}] : \BB(s) \]
  \begin{align*}
    \BB\left(f^{i,j}[\svec{e}] \!\leftarrow\! e_0 \right)
      &= \svec{e} \in ?^{\mem}f \land \svec{\BB_\cpu(e)} \land \BB_\mem(e_0)  \\
      &  \land \svec{e} \in ?^{\cpu}f^{i,j} \implies \BB_\cpu(e_0)
  \end{align*}
  \emph{\scriptsize where}
  \[ \BB_\cpu(f^{i,j}[\svec{e}]) = \svec{\BB_\cpu(e)} \land \svec{e} \in ?^{\mem}f \land \svec{e} \in ?^{\cpu}f^{i,j} \]
  \[ \BB_\mem(f[\svec{e}]) = \svec{\BB_\cpu(e)} \land \svec{e} \in ?^{\mem}f \]
  \emph{\scriptsize Remaining cases for $\BB_\cpu, \BB_\mem$ fold with union.}
  
  \caption{Query extraction function $p=\BB(T)$ produces predicate from program $T\in\kwd{Tgt}^?$.}
  \label{fig:bounds-extraction}
\end{subfigure}\hfill%
\begin{subfigure}[b]{0.43\textwidth}
\centering
\newcommand{\BB}[1]{\beta_0\!\left(#1\right)}
\newcommand{\II}[1]{\mathcal{I}\!\left(#1\right)}
\newcommand{\env}{\Gamma}
\begin{align*}
% \BB{\forall v\in I: b, \env} &= \BB{b, \env[v\mapsto [v, 1]]} [v\mapsto I] \\
\BB{\forall v\in I: b, \env} &= \BB{b, \env[v\mapsto I]} \\
\BB{b_1 \land b_2, \env} &= \BB{b_1, \BB{b_2, \env}} \\
\BB{e \implies b, \env} &= \BB{b, \env} \\
\BB{\exists~?^\ell : b, \env} &= \BB{b, \env} \\
\BB{\kwd{let} v = e \kwd{in} b, \env} &= \BB{b, \env[v\mapsto \II{e, \env}]} \\
\BB{\svec{e} \in ~?^\ell, \env} &= \env\left[?^\ell \mapsto \env(?^\ell) \cup \svec{\II{e, \env}}\right]
\end{align*}

\hrule
\vspace{1ex}

\fbox{Note: $[\omega^{\min}_i, \omega^{\len}_i] = \II{e_i}$ below.}
\[ \II{e_1 + e_2, \env} = [\omega^{\min}_1 + \omega^{\min}_2, \omega^{\len}_1 + \omega^{\len}_2 - 1] \]
\emph{(standard interval arithmetic rules elided)}
\begin{align*}
\II{e_1 \Div e_2, \env} &= [-M, 2M - 1] \\
\intertext{\emph{where $M=\max(-\omega^{\min}_1, \omega^{\min}_1+\omega^{\len}_1-1)$}}
\II{f^\ell[\svec{e}], \env} &= [-\infty, \infty] \\
\II{v, \env} &= \env(v) \\
\II{c, \env} &= [c, 1]
\end{align*}
\caption{Baseline bounds engine $\beta_0(p)$ applies naive interval arithmetic rules to queries produced by $\mathcal{B}$.}
\label{fig:bounds-inference}
\end{subfigure}
\end{framed}
\caption{Overview of the bounds inference system, showing query extraction and the baseline bounds engine $\beta_0$.}
\end{figure}

%% file: tex/scheduling.tex
We formalize scheduling by directly mutating programs in $\kwd{Tgt}^?$.
Because some directives --- like split --- must be applied after certain other directives, we require that schedules be ordered into phases\footnote{The practical system sorts directives into phases automatically.} as indicated in Figure~\ref{fig:schedule-syntax}.
Scheduling directives use loop names to determine their targets.

\input{figures/schedule/syntax}
In Figure~\ref{fig:schedule-semantics} we show the IR transformations for each scheduling directive.
In each subsequent section, we describe each phase and enumerate its restrictions, but defer safety proofs to the appendices.
For each phase, we require an inductive lemma like the following:

\begin{lemma}[Scheduling phase is sound]\label{lem:schedule-phase-generic}
Let $P\in\kwd{Alg}$ be a valid algorithm and let $T_i \in\kwd{Tgt}^?$ be the result of lowering $P$ and applying scheduling directives up through this phase.
Let $s$ be a scheduling directive in this phase, then $T_{i+1}=\mathcal{S}(s, T_i)$ is confluent with $P$.
\end{lemma}

\subsection{Specialization Phase}\label{sec:scheduling-specialization}
Certain scheduling decisions may be more or less efficient, depending on program parameters.
For instance, simpler schedules tend to work better for small output sizes.

\emph{Specialization} duplicates an existing func's code for each of $n$ conditions, and introduces labels that allow later scheduling directives to operate differently on each instance.
These conditions, like all expressions in the scheduling language, are required to be \emph{start-up expressions}.
In our formal system, schedules may give at most one specialization directive per func.
The following lemma captures an essential property of specializations, namely that only one specialization is ``active'' during any given run.

\begin{lemma}[Unique active specialization]\label{lem:specialize-active}
Given algorithm $P \in\kwd{Alg}$ and a schedule $S \in \kwd{Sched}$, let $P'\in\kwd{BI}(\mathcal{S}(S,\mathcal{L}(P)))$.
Then for any input $z$, $P'(z)$ will evaluate exactly one specialization for any given func $f$.
\end{lemma}

It is also important to note that the transformation attaches the specialization instance to func references inside the copied (and original) statements.
As per definition~\ref{def:bounds-hole}, the rule in Figure~\ref{fig:schedule-semantics} should be interpreted to exclude the final stage when attaching this information.

\subsection{Loops Phase}\label{sec:scheduling-loops}
Halide provides several standard loop transformations to change the order of computations.
A loop can be \emph{split} into two nested loops, two nested loops can be \emph{fused} into a single loop, a loop may be \emph{swapped} with the immediately nested loop, and loops may be traversed in \emph{parallel}.
Swapping and parallelization apply only to \emph{pure loops}, a manifestation of pure dimensions in the target IR.
We define these here:
\input{figures/schedule/semantics}

\begin{definition}[Pure loop]\label{def:pure-loop}
Let $v$ be the loop variable for some loop in a program $P\in\kwd{Tgt}^?$.
We say $v$ and its associated loop are \emph{pure} if $v$ is one of the pure dimensions of its associated func, if it is the result of splitting a pure loop, or if it is the result of fusing two pure loops together.
We letter the iteration variable of pure loops $x$. All other loops are \emph{reduction loops}, lettered $r$.
\end{definition}

We may $\kwd{split}$ a loop $\ell$ by a \emph{split factor} of $e$ into an outer loop iterated by $x_o$ and inner loop iterated by $x_i$.
This division may produce a remainder, which is handled by choice of a \emph{tail strategy} (denoted $\varphi$): (1) \emph{guard}ing the body with an $\kwd{if}$; (3) \emph{shift}ing the last loop iteration inwards, causing recomputation; or (3) \emph{round}ing the loop bounds upward, causing overcomputation of the func and affecting upstream bounds.
Shifting and rounding are only allowed on pure stages.

Two nested loops can be \emph{fused} together into a single loop whose extent is the product of the original extents, provided both loops are pure or both are reduction loops.
This is approximately an inverse to the split directive, and is useful for controlling the granularity of parallelism.
Immediately nested loops can be \emph{swapped} as long as the swap does not reorder two reduction loops.
Finally, each pure loop can also be traversed in either \emph{serial} or \emph{parallel} order.
All variable names introduced by these directives must be new, unique, and non-conflicting.

\subsection{Compute Phase}\label{sec:scheduling-compute}
To narrow the scope of computation, the $\kwd{label}$ed statement for computing a func $f$ may be moved from the top level to just inside any loop as long as the $\kwd{label}$ed statement continues to dominate all external accesses to $f$.

The closer a producer is computed to its consumer, the less of the producer needs to be computed per iteration of the consumer.
The expectation is that bounds inference will use the additional flexibility granted by the additional loop iteration information to derive tighter bounds.
This directive therefore controls \emph{how much} of a func to compute before computing part of its consumers.

\subsection{Storage Phase}\label{sec:scheduling-storage}
Each func is tied to a particular piece of memory when it is computed.
Halide offers some control over how much memory a func occupies during the run of a pipeline.
The \emph{store-at} directive (analogous to compute-at above) moves the allocation statement to just inside any loop such that the allocation still dominates all accesses of the func it allocates.

Bounds inference is then free to choose a more precise size for the allocation based on the code that follows, and the particular values of the variables of the loops that enclose it.

\subsection{Bounds Phase}\label{sec:scheduling-bounds}
Additional domain knowledge might allow a user to derive superior bounds functions than those inferred.
Halide provides directives to give hints to the bounds engine just before querying it.

The first two directives, \emph{bound} and \emph{bound-extent}, assert equality of bounds holes to provided startup expressions.
The third directive, \emph{align-bounds}, adds assertions that constrain the divisibility and position of the window.
The minimum is constrained to have a particular remainder modulo a factor which is declared to divide the extent.
These assertions affect the bounds inference query such that the inferred computation window will expand to meet these requirements.
Recall that these assertions are allowed to fail without violating \emph{confluence} (definition \ref{def:confluence}).

\subsection{Practical directives}\label{sec:schedule-practical}
Halide provides many more scheduling directives that are out of scope for this paper.
It has directives for assigning loops to coprocessors like GPUs and the DSPs, and directives for prefetching and memoization.
It has two additional traversal orders that apply only to constant-extent loops after bounds inference has completed and are semantically uninteresting: $\kwd{vectorize}$, which asks Halide to vectorize the loop, and $\kwd{unroll}$, which simply unrolls the loop.
Some of the most esoteric directives
may require more substantial adjustments to these semantics.

%% file: figures/schedule/syntax.tex
\begin{wrapfigure}[17]{R}{0.53\textwidth}
\footnotesize
\begin{tabularx}{\linewidth}{rcXl}
  \toprule
  $S$    & $::=$ & $s_1 ; \dots ; s_n$                       & schedule program \\
  $\ell$ & $::=$ & $\tup{f, i, j, v}$                        & loop names (\S\ref{lem:loop-naming}) \\
  $\tau$ & $::=$ & $\kwd{serial} ~|~ \kwd{parallel}$         & traversal orders \\
  $\varphi$ & $::=$ & $\varphi_\text{Guard} ~|~ \varphi_\text{Shift} ~|~ \varphi_\text{Round}$ & split strategies \\
  \toprule
  $s$    & $::=$ & $\kwd{specialize}(f, e_1, \dots, e_n)$    & Specialization (\S\ref{sec:scheduling-specialization}) \\
  \midrule
         & |     & $\kwd{split}(\ell, x_o, x_i, e, \varphi)$ & Loops (\S\ref{sec:scheduling-loops})\\
         & |     & $\kwd{fuse}(\ell, x)$                     & \\
         & |     & $\kwd{swap}(\ell)$                        & \\
         & |     & $\kwd{traverse}(\ell, \tau)$              & \\
      %  & |     & $\kwd{unroll}(\ell)$                      & \\
  \midrule
         & |     & $\kwd{compute-at}(f, \ell_g)$             & Compute (\S\ref{sec:scheduling-compute}) \\
  \midrule
         & |     & $\kwd{store-at}(f, \ell_g)$               & Storage (\S\ref{sec:scheduling-storage}) \\
      %  & |     & $\kwd{fold-storage}(f, x, e)$             & \\
      %  & |     & $\kwd{sliding-window}(f, x)$              & \\
  \midrule
         & |     & $\kwd{bound}(f, x, e^{\min}, e^{\len})$   & Bounds (\S\ref{sec:scheduling-bounds}) \\
         & |     & $\kwd{bound-extent}(f, x, e^{\len})$      & \\
         & |     & $\kwd{align-bounds}(f, x, e^m, e^r)$      & \\
  \bottomrule
\end{tabularx}
\Description{Scheduling language}
\caption{The Halide scheduling language. The $s$ definition is grouped by phase of scheduling (presented in order).}
\label{fig:schedule-syntax}
\end{wrapfigure}

%% file: figures/schedule/semantics.tex
\begin{figure}
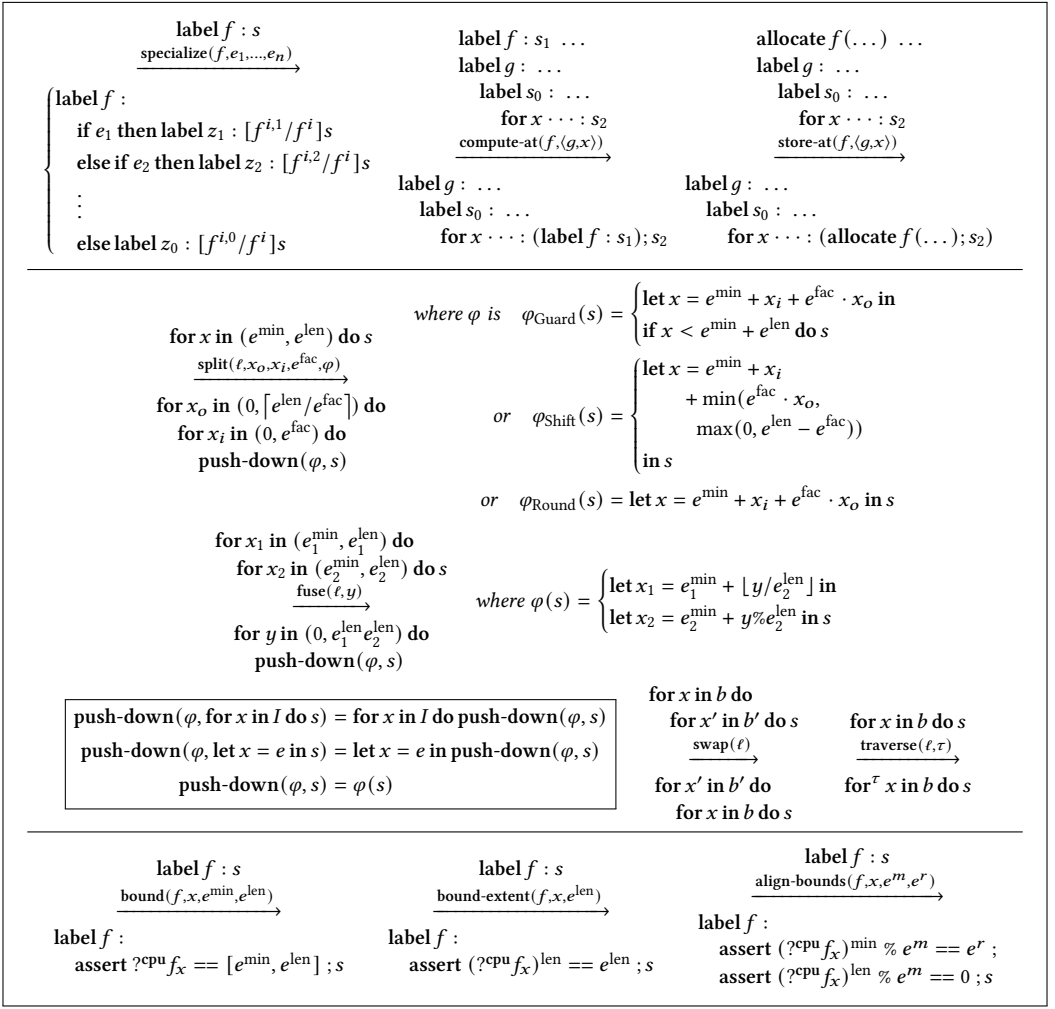

\footnotesize
\begin{framed}
\vspace*{-1em}
  \[
  %%%
  \begin{array}{c}
    \kwd{label} f : s \\
    \xrightarrow{\kwd{specialize}(f, e_1, \dots, e_n)} \\
    \begin{cases}
      \kwd{label} f : & \\
      \quad \kwd{if} e_1 \kwd{then} \kwd{label} z_1 : [f^{i,1}/f^i] s & \\
      \quad \kwd{else} \kwd{if} e_2 \kwd{then} \kwd{label} z_2 : [f^{i,2}/f^i] s & \\
      \quad \vdots & \\
      \quad \kwd{else} \kwd{label} z_0 : [f^{i,0}/f^i] s &
    \end{cases}
  \end{array}%
  \hspace{-4ex}%
  \begin{array}{c}
    \begin{array}{l}
      \kwd{label} f : s_1 \; \dots \\
      \kwd{label} g : \; \dots \\
      \quad \kwd{label} s_0 : \; \dots \\
      \quad \quad \kwd{for} x \dots : s_2
    \end{array} \\
    \xrightarrow{\kwd{compute-at}(f, \tup{g, x})} \\
    \begin{array}{l}
      \kwd{label} g : \; \dots \\
      \quad \kwd{label} s_0 : \; \dots \\
      \quad \quad \kwd{for} x \dots : (\kwd{label} f : s_1) ; s_2
    \end{array}
  \end{array}%
  \hspace{-4ex}%
  \begin{array}{c}
    \begin{array}{l}
      \kwd{allocate} f (\dots) \; \dots \\
      \kwd{label} g : \; \dots \\
      \quad \kwd{label} s_0 : \; \dots \\
      \quad \quad \kwd{for} x \dots : s_2
    \end{array} \\
    \xrightarrow{\kwd{store-at}(f, \tup{g, x})} \\
    \begin{array}{l}
      \kwd{label} g : \; \dots \\
      \quad \kwd{label} s_0 : \; \dots \\
      \quad \quad \kwd{for} x \dots : (\kwd{allocate} f (\dots) ; s_2)
    \end{array}
  \end{array}
  \]
  %%%
  \hrule
  %%%
  \[
  \begin{array}{c}
    \kwd{for} x \kwd{in} ~(e^\mathrm{min}, e^\mathrm{len}) \kwd{do} s \\
    \xrightarrow{\kwd{split}(\ell, x_o, x_i, e^\mathrm{fac}, \varphi)} \\
    \begin{array}{l}
      \kwd{for} x_o \kwd{in} ~(0, \ceil*{e^\mathrm{len} / e^\mathrm{fac}}) \kwd{do} \\
      \quad \kwd{for} x_i \kwd{in} ~(0, e^\mathrm{fac}) \kwd{do} \\
      \quad \quad \kwd{push-down}(\varphi, s)
    \end{array}
  \end{array}
  \begin{aligned}
  \text{\emph{where $\varphi$ is}}\quad%
  \varphi_{\mathrm{Guard}}(s) &= \begin{cases}
    \kwd{let} x = e^\mathrm{min} + x_i + e^\mathrm{fac} \cdot x_o \kwd{in} & \\
    \kwd{if} x < e^\mathrm{min} + e^\mathrm{len} \kwd{do} s
  \end{cases} \\
  \text{\emph{or}}\quad%
  \varphi_{\mathrm{Shift}}(s) &= \begin{cases}
    \begin{aligned}
      \kwd{let} x &= e^{\min} + x_i \\
                  &+ \min(e^{\mathrm{fac}}\cdot x_o, & \\
                  &\phantom{+} \max(0, e^{\len} - e^{\mathrm{fac}}))
    \end{aligned} & \\
    \kwd{in} s &
  \end{cases} \\
  \text{\emph{or}}\quad%
  \varphi_{\mathrm{Round}}(s) &= \kwd{let} x = e^\mathrm{min} + x_i + e^\mathrm{fac} \cdot x_o \kwd{in} s
  \end{aligned}
  \]
  %%%
  \[
  \begin{array}{c}
    \begin{array}{l}
      \kwd{for} x_1 \kwd{in} ~(e^\mathrm{min}_1, e^\mathrm{len}_1) \kwd{do} \\
      \quad \kwd{for} x_2 \kwd{in} ~(e^\mathrm{min}_2, e^\mathrm{len}_2) \kwd{do} s
    \end{array} \\
    \xrightarrow{\kwd{fuse}(\ell, y)} \\
    \begin{array}{l}
      \kwd{for} y \kwd{in} ~(0, e^\mathrm{len}_1 e^\mathrm{len}_2) \kwd{do} \\
      \quad \kwd{push-down}(\varphi, s)
    \end{array}
  \end{array}
  \text{\emph{where}}\;
  \varphi(s) = \begin{cases}
    \kwd{let} x_1 = e^\mathrm{min}_1 + \floor{y / e^\mathrm{len}_2} \kwd{in} & \\
    \kwd{let} x_2 = e^\mathrm{min}_2 + y \% e^\mathrm{len}_2 \kwd{in} s & \\
  \end{cases}
  \]
  \[
  \fbox{$\begin{aligned}
    \kwd{push-down}(\varphi, \kwd{for} x \kwd{in} I \kwd{do} s)
      &= \kwd{for} x \kwd{in} I \kwd{do} \kwd{push-down}(\varphi, s)\\
    \kwd{push-down}(\varphi, \kwd{let} x = e \kwd{in} s)
      &= \kwd{let} x = e \kwd{in} \kwd{push-down}(\varphi, s)\\
    \kwd{push-down}(\varphi, s) &= \varphi(s) \\
  \end{aligned}$}
  \;
  \begin{array}{c}
    \begin{array}{l}
      \kwd{for} x \kwd{in} b \kwd{do} \\
      \quad \kwd{for} x' \kwd{in} b' \kwd{do} s
    \end{array} \\
    \xrightarrow{\kwd{swap}(\ell)} \\
    \begin{array}{l}
      \kwd{for} x' \kwd{in} b' \kwd{do} \\
      \quad \kwd{for} x \kwd{in} b \kwd{do} s
    \end{array} \\
  \end{array}
  \;
  \begin{array}{c}
    \kwd{for} x \kwd{in} b \kwd{do} s \\
    \xrightarrow{\kwd{traverse}(\ell, \tau)} \\
    \kwd{for}^\tau x \kwd{in} b \kwd{do} s
  \end{array}
  \]
  %%%
  \hrule
  %%%
  \[
  \begin{array}{c}
    \kwd{label} f : s \\
    \xrightarrow{\kwd{bound}(f, x, e^{\min}, e^{\len})} \\
    \begin{array}{l}
      \kwd{label} f : \\
      \quad \kwd{assert} ~~ ?^{\kwd{cpu}}f_x == [e^{\min}, e^{\len}] ~; s
    \end{array}
  \end{array}
  \hspace{-1ex}
  \begin{array}{c}
    \kwd{label} f : s \\
    \xrightarrow{\kwd{bound-extent}(f, x,e^{\len})} \\
    \begin{array}{l}
      \kwd{label} f : \\
      \quad \kwd{assert} ~~ (?^{\kwd{cpu}}f_x)^{\len} == e^{\len} ~; s
    \end{array}
  \end{array}
  \hspace{-1ex}
  \begin{array}{c}
    \kwd{label} f : s \\
    \xrightarrow{\kwd{align-bounds}(f, x, e^m, e^r)} \\
    \begin{array}{l}
      \kwd{label} f : \\
      \quad \kwd{assert} ~~ (?^{\kwd{cpu}}f_x)^{\min} ~\%~ e^m == e^r ~; \\
      \quad \kwd{assert} ~~ (?^{\kwd{cpu}}f_x)^{\len} ~\%~ e^m == 0 ~; s
    \end{array}
  \end{array}
  \]
\vspace*{-1em}
\end{framed}
\Description{Scheduling directives over the IR}
\caption{Scheduling directives over the IR}
\label{fig:schedule-semantics}
\end{figure}

%% file: tex/related_work.tex
The computational and scheduling models of Halide have evolved through a series of extensions and generalizations \cite{Halide:PLDI, Halide:SIGGRAPH, Halide:CACM, Halide:rfactor}.
Halide builds on the idea of explicit control over compiler transformations developed earlier in many script- or pragma-based compiler tools in HPC \cite{Xlang, Sequoia, POET, Orio, CHiLL}, and the definition of parametric spaces of optimizations in SPIRAL \cite{SPIRAL}.
A growing family of high performance DSLs since the introduction of Halide have directly adopted the concept of a programmer-visible scheduling language \cite{Legion,TVM,TensorComprehensions,TACO,GraphIt,SWIRL,Taichi}.
The polyhedral loop optimization community has explored similar ideas in its own context \cite{ISL,ScheduleTrees,PENCIL,Tiramisu}.

Virtually all of these languages and systems do not have formally specified semantics, proofs of soundness, or other such metatheory. 
POET~\cite{POET} and TeML~\cite{TeML} are notable exceptions for being defined formally, but their transformation (i.e., scheduling) languages are not shown to be correctness-preserving.
Legion defined a core calculus and proved a form of soundness for their dynamic, user-configurable distributed scheduler~\cite{Legion-Partition-2013}.
However, for our present aims many of the details are unnecessary, and redundant recomputation and overcomputation on uninitialized values---both essential to Halide---remain outside their scope.
Egg~\cite{Egg}, ELEVATE~\cite{ELEVATE}, and the X language~\cite{Xlang} all provide generic transformation or rewriting infrastructure, but do not provide the definitions and metatheory needed to establish correctness for any specific language.

URUK~\cite{URUK2005,URUK2006,URUK2007}, CHiLL~\cite{CHiLL}, and Tiramisu~\cite{Tiramisu} are notable examples of user-schedulable polyhedral compilers.
In the latter two cases, correctness claims are deferred to polyhedral dependence analysis using ISL~\cite{ISL}.
As we discuss in \S\ref{sec:overview}, dependence analysis is only sufficient to justify \emph{re-ordering} transformations---not transformations such as $\kwd{compute-at}$, which \emph{recompute} or \emph{over-compute} values and might introduce novel statement instances.
For instance, in correspondence with authors of the Tiramisu paper and system~\cite{Tiramisu-correspondence}, we discovered that the relevant safety checks for $\kwd{compute-at}$ transformations had neither been implemented in the system artifact, nor described in the paper.  

Older automated polyhedral analyses~\cite{Feautrier91} work on static control programs with denotational / functional semantics.
In that setting, dataflow and dependence graphs are equivalent.
This is also the case for functional DSLs such as PolyMage~\cite{PolyMage}, which also supports redundant recomputation.
Recent developments in the Alpha system~\cite{AlphaZ} are notable for maintaining a complete denotational form of the program throughout transformation, not just a dependence analysis.
As with Halide, functional semantics are crucial for reasoning about such non-re-ordering code transformations.

Halide's algorithm language is closely related to both array languages \cite{APL,ZPL,Chapel,Futhark:system,HaskellAccelerate}, and image processing DSLs \cite{Popi,Shantzis}.
Its computational model is most closely related to that of the lazy functional image language Pan \cite{Pan}.
Bounds inference is related to array shape analyses and type systems \cite{ShapeCheckingArrayPrograms,FunctionalImperativeShape,Futhark:bounds}.
Our treatment of bounds inference is (to the best of our knowledge) the first formulation via a constraint-based program synthesis problem \cite{ProgramSynthesis}.

The correctness of many compiler transformations has been treated in the context of verified compilers like Comp\-Cert \cite{CompCert:dataflow,CompCert:isched,CompCert:vliw}.
The closest component to the present work is the Comp\-Cert instruction-scheduling optimization, which is designed to be applied after register allocation.
(By contrast, we are concerned with less local and harder-to-validate loop transformations.)
Verification is based on the \emph{translation validation} strategy, where a certified validator program attempts to prove that the pre- and post-optimization programs are equivalent.
This strategy is effective in the Comp\-Cert scenario because (a) it is (potentially) generic with respect to the choice of optimization pass and (b) when validation fails, Comp\-Cert can always (correctly) fall back to a less optimized version of the code.
Once scheduling is exposed to the user (our scenario), these design choices are inappropriate.
The semantics must make predictable and defensible guarantees to users about the results of schedules that they write.

Concurrent work by \citet{Halide-TRS} uses program synthesis to build a verified term-rewriting expression simplifier for the Halide expression language.
Their verification conditions are based on the expression language semantics described in this work.

\JRK{Add \cite{array-dependent-types}, Hughes paper}

%% file: tex/impact.tex
These formalization efforts have influenced Halide's design, and we have found and fixed bugs where actual and expected behavior differed in significant ways.

\paragraph{Negative rdom extents.}
While formalizing the behavior of reduction domains (\S\ref{sec:alg_lang-semantics}), we discovered that the practical system had not defined the behavior of loops with \emph{negative} extents \cite{bug-neg-rdom}.
Test cases designed to probe the behavior suggested that Halide treated such loops as no-ops; however, there could be instances wherein a negative extent is treated as unsigned, which would silently wrap to a very large positive integer.
While unsigned underflow is a well-known problem, Halide has the additional obligation of making sure that no scheduling transforms accidentally introduce this behavior even if it’s absent in the original code.
We worked with the developers to determine that this situation should be treated as an error that can be checked at program startup (recall definition \ref{def:startup-expression}), as formalized here.

\paragraph{Impure identity functions.}
The practical system has several APIs for computing the results of a pipeline.
One such API intended to match the interface formalized here (\S\ref{sec:alg_lang}, \S\ref{sec:target_lang}): the user supplies the desired \emph{compute bounds} (see \S\ref{sec:bounds}) and receives a buffer containing at least the requested values.

For efficiency, another API allows a user to supply their own output buffer, rather than delegating the allocation to Halide.
In this case, the pipeline checks at startup that the supplied buffer is at least as large as the buffer it \emph{would have} allocated.
However, when the simple API was implemented in terms of this advanced API, it incorrectly assumed that the \emph{compute bounds} and \emph{allocation bounds} would be equal.
This led to vexing errors on pipelines whose outputs were scheduled to overcompute \cite{bug-output-RDom}.

This confusion had a surprising consequence: adding an unscheduled \emph{identity} func to the end of the pipeline would compile to a \emph{copy} of the former output, and which would have equal compute and allocation bounds.
So, from the perspective of the user, identity functions were \emph{impure} since they had side effects due to bounds inference.
After reaching clarity on these issues through our formalism, we worked with the Halide authors to fix this behavior.
The latest release correctly returns the full, possibly overallocated, buffer.

\paragraph{Arithmetic error semantics.}
As discussed in \S\ref{sec:bounds-extraction}, values used in control flow or indexing must be \emph{well-defined} regardless of the schedule.
Data-dependent accesses in rdom conditions and update locations might necessitate computing points not required by the default schedule, especially when over-computing strategies are employed.
Similarly, computation outside the compute bounds must be side-effect free, even when processing uninitialized values.
One consequence of this is that integer division and modulo must be made into \emph{total} functions, similar to IEEE 754 arithmetic.

We constructed test cases for the practical system that crashed due to integer division by zero happening outside of the compute bounds \cite{bug-roundup}.
We worked with the Halide authors to \emph{define} these operations to return zero and implemented the new behavior with runtime checks.
The compiler leverages its existing bounds analyses to eliminate these checks when it can, for instance when dividing by a non-zero constant.

One might wonder why the convention $x \bmod 0 \equiv 0$ was chosen in favor of the more typical $x \bmod 0 \equiv x$. Both conventions were tried and the former produced tighter bounds equations in practice; in short, it is better to bound $(x\bmod y)$ by $y$ than by $x$, which is typically much wider.

This change also impacted concurrent work on verifying Halide's term-rewriting expression simplifier.
As \citet{Halide-TRS} report, these new semantics invalidated dozens of existing rewrite rules and required many new proofs of correctness for valid rules.

\paragraph{Compute bounds for indexing accesses.}
Another consequence of the rules in \S\ref{sec:bounds-extraction} is that accesses that occur inside indexing expressions must have well-defined values, which means that the points must be in the compute bounds.
However, the practical system did not implement this rule; it instead relied on an unsound analysis of the bounds of a func's value to compute the bounds in the indexing expression and did not widen the compute bounds to fit the accessed point.
We were able to construct a real crash based on this insight in \cite{bug-nested-access-bounds} and provided a patch to the compiler.

\paragraph{Race conditions in rdom predicates}
It is unsafe to parallelize a loop that contains an RDom whose predicate depends on
values written by that loop. Race conditions on the values read by the predicate can
lead to non-deterministic behavior. We discovered that the compiler was missing these 
checks. We constructed a real instance of non-determinism based on this insight
\cite{bug-rdom-race} and provided a patch to the compiler\cite{bug-rdom-race-patch}.

\paragraph{Compute-with directive.}
\emph{Compute-with} was a scheduling directive intended to interleave the computation of two or more independent funcs by fusing their outermost loops together.
This could benefit performance by reducing memory traffic if the two funcs shared many reads from a common producer.
However, the prototype implementation did not consider dependencies between the stages of a single func (\S\ref{sec:alg_lang-semantics}), nor did it consider \emph{specializations} (\S\ref{sec:scheduling-specialization}).
We discovered cases where compute-with could move the pure stage of a func \emph{after} one of its update stages, resulting in crashes and mangled outputs \cite{bug-compute-with}.

We worked with the Halide authors to define the feature, but due to little widespread use (perhaps owing to these bugs, in part) and the highly complex implementation, the feature was deprecated instead.
We look forward to designing a sound replacement in future work.

\paragraph{Future work}
We believe this work provides a foundation to study the new class of languages with user-controlled scheduling. One major question is how they could incorporate abstraction and module systems.
Another is whether alternative bounds inference algorithms, based on our program synthesis formulation, could be useful in practice and in other settings.

%% file: tex/acknowledgments.tex
We thank Andrew Adams and Daan Leijen for their helpful conversations about Halide's implementation and language semantics, respectively.
We also thank Martin Rinard, Zachary Tatlock, Adam Chipala, and Sarah Chasins for their detailed review and feedback prior to submission.

%% file: tex/appendices.tex
\section{Proofs of Theorems and Lemmas}\label{app:proofs}

\begin{lemma}[Loop phase naming]\label{lem:loop-preserves-names}
The loop phase preserves loop names as described in Lemma~\ref{lem:loop-naming}.
\end{lemma}

\begin{proof}
All loop phase directives replace loops of a given func, specialization, and stage with similarly localized loops, albeit with different loop variables.
Since conflicts between loop variable names are prohibited, loops remain uniquely named.
\end{proof}

\begin{definition}[Narrowing program]\label{def:narrowing}
Let $P_1, P_2\in\kwd{BI}(T)$ where $T$ is a program scheduled from some algorithm.
We say that $P_1 \leq P_2$ iff for all inputs $z$, every execution of any allocation or loop in $P_1(z)$ has a corresponding execution in $P_2(z)$ and the bounds $I_1$ for $P_1(z)$ are contained in the bounds $I_2$ for $P_2(z)$.
(Here ``corresponding'' means that the two statements are at the same lexical site, executing with identical environments $\Sigma$.)
\end{definition}

\begin{lemma}[Narrowing executions match]\label{lem:narrowing-match}
Let $P_1 \leq P_2$ as above.
Then after the execution of corresponding loops for some func $f$ on intervals $I_1$ and $I_2$, the contents of the buffer for $f$ agree on $I_1$.
\end{lemma}

\begin{proof}[Proof of Lemma~\ref{lem:narrowing-match}]
Let $\kwd{SI}(P, z)$ denote the set of assignment executions in $P(z)$.
By definition, $\kwd{SI}(P_1, z) \subseteq \kwd{SI}(P_2, z)$.
Let $\kwd{SI}(P_2, z)|_{I_1}$ be the set of assignment statement executions in $P_2(z)$ which write which write to a point of $f$ in $I_1$.
Then $\kwd{SI}(P_1, z) \supseteq \kwd{SI}(P_2, z)|_{I_1}$. %(by the $\kwd{CPU}\subseteq \kwd{Loop}$ constraint of $\mathcal{B}$).
Thus there are no writes to values in $I_1$ in $P_2$ that do not also occur in $P_1$.
Furthermore, all of these computed values depend only on the values computed in both programs.
This follows from the assignment rule of $\mathcal{B}$.
\end{proof}

For the next three proofs, note that if an expression appearing in an access is unbounded then there is no possible bounds query result.
So without loss of generality, we may assume that some satisfying bounds \emph{actually do exist}, since confluence is vacuous otherwise.

\begin{proof}[Proof of Lemma~\ref{lem:bounds-memory-safety}]
First recall that by Lemma~\ref{lem:dominance}, every access to a func $f$ is dominated by the $\kwd{allocate}$ statement for $f$.
The rule for accesses in $\mathcal{B}$ (see Figure \ref{fig:bounds-extraction}) explicitly requires that every point read or written is contained in the allocation bounds ($\kwd{Mem}$) of the associated func.
Thus $P'$ will always satisfy the \textsc{InBounds} condition.
\end{proof}

\begin{proof}[Proof of Lemma~\ref{lem:compute-confluency}]
This proof proceeds by induction over prefixes of the syntactic structure of the algorithm $P\in\kwd{Alg}$ and corresponding prefixes of the initially lowered program $P'\in\kwd{BI}(\mathcal{L}(P))$.

Recall from the definition of the bounds extraction function $\mathcal{B}$ (\S\ref{sec:bounds-extraction}) that each of $f$'s stages' loop bounds cover the compute bounds.
Lemma \ref{lem:dominance} ensures that the func is realized before it is read.

Thus, as a base case, if $f$ consists of only a pure stage, we are done because the static lack of self-reference and reduction dimensions makes each assignment statement in that stage completely independent of every other (this is ensured by Definition \ref{def:separation}).
The assignment exactly matches the algorithm's expressions and only reads from compute bounds by assumption.

Inductively, if $f$ has $n$ stages the claim holds, we argue that adding another stage $s$ to $f$ preserves the claim.
Bounds extraction ensures that every access in $s$ is included in the allocation bounds and assignments writing to a point in the compute bounds read only within the compute bounds of other funcs and $f$.
By the induction hypothesis on stages, the values in the compute bounds of the buffer for $f$ (and all other funcs) are correct just before $s$ runs.

Recall the structure of the loop nest for $s$.
The outermost for loops correspond to pure dimensions and range over bounds holes, which are constrained to cover the compute bounds of $f$.
The innermost for loops implement any reduction domain in the stage and have bounds supplied by the algorithm.
If the stage is guarded by a condition, it is included within the innermost for loop.
Finally, a single assignment statement corresponding to the update rule for the stage is the innermost statement.

We need to show that the code produced by lowering for $s$ will compute confluent values for $f$.
Consider a point $p$ in the compute bounds of $f$ after the stage runs.
The stage $s$ separates pure dimensions from reduction dimensions of $p$.
Let $x_p$ be the assignment to pure dimensions induced by $p$.
Consider the definition of $s$ in the algorithm.
It consists of a series of simple updates after unrolling the rdom all of which share values $x_p$ in common.
Now, observe the iteration of the pure loops of $s$ in the target language which coincide with $x_p$.
This iteration is guaranteed to occur by the covering of the compute bounds by the loop bounds.
The content of this iteration is a sequence of assignments corresponding to the unrolled rdom (as in the big step semantics in Figure \ref{fig:source-semantics}).
Lastly, any other $x_p' \neq x_p$ touches no memory in common with this iteration because of syntactic separation (Definition \ref{def:separation}).
\end{proof}

\begin{lemma}[Lowering is sound]\label{lem:lowering-sound}
$P'\in\kwd{BI}(\mathcal{L}(P))$ is confluent with $P$ for all $P\in\kwd{Alg}$.
\end{lemma}

\begin{proof}[Proof of Lemma~\ref{lem:lowering-sound}]
Let $z$ be any input. If $P(z)$ contains an error, then we are done.
Since lowering does not create any assertion statements, failing one is not possible.
All lowered programs trivially respect dominance.
Finally the argument in lemma \ref{lem:compute-confluency} applies inductively over the lowered code.
\end{proof}

\begin{proof}[Proof of Lemma~\ref{lem:specialize-active}]
If $f$ has no specializations, then its only compute statement is considered the default specialization.
Valid schedules have at most one specialization directive per func, so there is one block of if-then statements for each func, each containing a compute statement for $f$.
Since the conditions of those branches are startup expressions, the input $z$ determines which one will be taken every time the block is encountered.
Proofs that each subsequent phase of scheduling preserves this invariant appear in appendix \ref{app:proofs}.
\end{proof}

\begin{lemma}[Specialization is sound]\label{lem:schedule-specialize}
Let $P\in\kwd{Alg}$ be a valid algorithm and let $T_i \in\kwd{Tgt}^?$ be the result of lowering $P$ and applying scheduling directives up through this phase.
Let $s$ be a scheduling directive in this phase, then $T_{i+1}=\mathcal{S}(s, T_i)$ is confluent with $P$.
\end{lemma}

\begin{proof}[Proof of Lemma~\ref{lem:schedule-specialize}]
Let $z$ be a valid input for $P$, $P_i' \in\kwd{BI}(T_i)$, and $P_{i+1}' \in\kwd{BI}(T_{i+1})$.
If $P(z)$ contains an error, we are done; and no assertions are present until the bounds phase; so we may assume $P \simeq_z P_i'$.
By assumption, $s$ is a specialization of some func $f$.
Now by lemma \ref{lem:specialize-active}, exactly one branch of $f$ is taken in any execution $P_{i+1}'(z)$.
This preserves producer domination as required by lemma \ref{lem:dominance}.
\end{proof}

\begin{lemma}[Loop phase is sound]\label{lem:schedule-loop}
Let $P\in\kwd{Alg}$ be a valid algorithm and let $T_i \in\kwd{Tgt}^?$ be the result of lowering $P$ and applying scheduling directives up through this phase.
Let $s$ be a scheduling directive in this phase, then $T_{i+1}=\mathcal{S}(s, T_i)$ is confluent with $P$.
\end{lemma}

\begin{proof}[Proof of Lemma~\ref{lem:schedule-loop}]
Lemma \ref{lem:schedule-specialize} ensures that $T_{i}$ is confluent as long as we have not yet reached this phase.
So we may inductively assume confluence before issuing any such $s$.
As before, $s$ does not introduce assertions, so we need only assess output equivalence.
Each directive $s$ operates locally on loops in a single stage, so therefore we need only show that the transformation of this stage is observationally equivalent; i.e. the state of the buffers before and after the stage is the same within the compute bounds.

Suppose $s=\kwd{split}(\tup{\dots,v}, v_o, v_i, e^{\mathrm{fac}})$. 
We may check that the values $v$ ranges over are unchanged; this means that $\mathcal{B}(T_i)$ is logically equivalent to $\mathcal{B}(T_{i+1})$.
Since the scopes of holes have also not changed ($K(T_i)=K(T_{i+1})$), the sets of bounds query results are identical. Thus there is a bijection between programs $P_i'\in\kwd{BI}(T_i)$ and $P_{i+1}'\in\kwd{BI}(T_{i+1})$, s.t. corresponding holes have been filled with identical expressions.
These two programs execute the same statement instances in the same order; therefore the effect on buffers is identical.

The argument for the $\kwd{fuse}$ directive proceeds analogously, as do the arguments for the overcomputation and inward-shifting split strategies that apply only to single-stage (pure) funcs, which are adjusted for the expanding compute bounds and identically recomputed points, respectively.

For $\kwd{swap}$ and $\kwd{traverse}$, consider the loop nest for the stage of $f$.
Analogously to the proof of lemma \ref{lem:compute-confluency}, syntactic separation ensures that distinct iterations of pure loops access $f$ at disjoint sets of points.
Since $\kwd{split}$ and $\kwd{fuse}$ preserve the sets and order of accesses as they introduce new loops, two writes agree on their pure dimensions iff the pure loop variable values agree.

But this means that two assignments touch the same memory of $f$ only if the values of the pure loops are the same.
Since $\kwd{traverse}$ parallelizes over a single pure loop, no two tasks touch the same memory.
And $\kwd{swap}$ interchanges two pure loops, so the writes which do touch memory in common are not interchanged with respect to one another.
\end{proof}

\begin{lemma}[Compute phase is sound]\label{lem:compute-sound}
Let $P\in\kwd{Alg}$ be a valid algorithm and let $T_i \in\kwd{Tgt}^?$ be the result of lowering $P$ and applying scheduling directives up through this phase.
Let $s$ be a compute-at directive moving the func $f$ to some location $\ell_g$. Then $T_{i+1}=\mathcal{S}(s, T_i)$ is confluent with $P$.
\end{lemma}

\begin{proof}[Proof of Lemma~\ref{lem:compute-sound}]
Let $P_i'\in\kwd{BI}(S_i)$ and let $\set{P_i''}$ be the set of programs resulting from applying $s$ to $P_i'$.
Let $P_{i+1}'\in\kwd{BI}(T_{i+1})$.
First observe that any $P_i''$ computes the same values as $P_i'$ because the $\kwd{compute}$ statement for $f$ computes \emph{all} of the points ever needed in $P_i''$ and dominates all of its consumer funcs.
This also implies that $P_i'' \in\kwd{BI}(T_{i+1})$.

Now suppose $P_{i+1}'$ is not one of the $P_i''$.
Then for each $z$, if $P(z)$ does not contain an error, then there is some $P_i'' \geq P_{i+1}'$ (see definition \ref{def:narrowing}).
By lemma \ref{lem:narrowing-match}, $P_{i+1}'$ computes the exact same values of $f$ on the narrower range, which bounds inference guarantees is sufficient for confluence.
\end{proof}

\begin{lemma}[Storage phase is sound]\label{lem:storage-sound}
Let $P\in\kwd{Alg}$ be a valid algorithm and let $T_i \in\kwd{Tgt}^?$ be the result of lowering $P$ and applying scheduling directives up through this phase.
Let $s$ be a store-at directive, then $T_{i+1}=\mathcal{S}(s, T_i)$ is confluent with $P$.
\end{lemma}

\begin{proof}[Proof of Lemma~\ref{lem:storage-sound}]
Lemma \ref{lem:bounds-memory-safety} required only dominance of the $\kwd{allocate}$ statement over all accesses to the associated func for correctness.
This is preserved by definition.
\end{proof}

\begin{lemma}[Bounds phase is sound]\label{lem:schedule-bounds}
Let $P\in\kwd{Alg}$ be a valid algorithm and let $T_i \in\kwd{Tgt}^?$ be the result of lowering $P$ and applying scheduling directives up through this phase.
Let $s$ be a scheduling directive in this phase, then $T_{i+1}=\mathcal{S}(s, T_i)$ is confluent with $P$.
\end{lemma}

\begin{proof}[Proof of Lemma~\ref{lem:schedule-bounds}]
The directives in this phase only mutate programs by adding assertions to them, the only side effect of which is to transition to an error state.
Thus in any non-erroring execution of any $P_{i}\in\kwd{BI}(S_i)$, there is some $P_{i}\in\kwd{BI}(T_i)$ whose behavior exactly matches $P_{i+1}$.
\end{proof}